\documentclass[10pt,journal,compsoc]{IEEEtran}
\ifCLASSOPTIONcompsoc
  \usepackage[nocompress]{cite}
\else
  \usepackage{cite}
\fi
\usepackage{mathrsfs,amssymb,stmaryrd,amsmath,mathtools,algorithm}
\usepackage{algorithmicx}
\usepackage{algpseudocode}
\usepackage{color}
\usepackage{tabularx}
\usepackage{verbatim}
\usepackage{tabulary}
\usepackage{longtable}
\usepackage{enumitem}
\usepackage{bm}
\usepackage{graphicx}

\def\si{-\!\!\!*~}

\def\dom{\mathrm{dom}}
\newtheorem{theorem}{Theorem}[section]
\newtheorem{proof}{Proof}[section]
\newtheorem{lemma}{Lemma}[section]

\ifCLASSINFOpdf
\else
\fi
\begin{document}
%
\title{Local Reasoning about Block-based \\ Cloud Storage Systems}
%
%
%
%

\author{Zhao Jin,~Hanpin Wang,~Lei Zhang,~Bowen Zhang,~and~Yongzhi Cao,
	\IEEEcompsocitemizethanks{\IEEEcompsocthanksitem Z. Jin, H. Wang, L. Zhang, B. Zhang, Y. Cao are with the Key Laboratory of High Confidence Software Technologies (MOE), School of Electronics Engineering and Computer Science, Peking University, Beijing 100871, China.
		\protect\\
		E-mail: \{jinzhao, whpxhy, zhangleijuly, zhangbowen, caoyz\}@pku.edu.cn.
		\IEEEcompsocthanksitem H. Wang is with the School of Computer Science and Cyber Engineering, Guangzhou University, Guangzhou 510006, China.}
	}

\IEEEtitleabstractindextext{%
\begin{abstract}
Owing to the massive growth in the storage demands of big data, Cloud Storage Systems (CSSs) have been put forward to improve the storage capacity.
Compare with traditional storage systems, CSSs have lots of advantages, such as higher capacity, lower cost, and easier scalability.
However, they suffer from the main shortcoming of high complexity.
To ensure the reliability of CSSs, the correctness of management programs should be proven. 
Therefore, a verification framework based on Separation Logic (SL) is proposed to prove the correctness of management programs in Block-based Cloud Storage Systems (BCSSs), which is the most popular CSSs. 
The main contributions are as follows.
(1) Two-tier heap structure is constructed as the type of storage units in BCSSs. All the operations to BCSSs are based on the structure. 
(2) Assertion pairs are defined to describe the properties for the two-tier structure.
The fact that the two components of a pair effect each other leads lots of interesting properties.
(3) A proof system with Hoare-style specification rules is proposed to reason about the BCSSs.
The results show that the correctness of BCSSs can be verified precisely and flexibly.
\end{abstract}

}

\maketitle

\IEEEdisplaynontitleabstractindextext

%
\IEEEpeerreviewmaketitle

\IEEEraisesectionheading{\section{Introduction}\label{sec:introduction}}

%
%
%
%
\IEEEPARstart{W}{ith} the rapid growth of data, the capacity of traditional storage devices could not meet the great demands.
The Cloud Storage Systems (CSSs) have been put forward to improve the storage capacity \cite{hashem2015rise}.
According to the data types stored, CSSs are divided into three categories: Block-based CSSs (BCSSs), Object-based CSSs, and File-based CSSs. Among these CSSs, BCSSs have the lowest cost and are the easiest to scale, hence BCSSs are the most popular systems used in CSSs at present.

In BCSSs, data are stored in the block-based storage structure, which means that the resources in BCSSs consist of small block spaces. When user submit their data file to a BCSS, the system would cut the file into some segments, and then take those segments into property blocks.
Using Hadoop distribeted file system (HDFS) \cite{shvachko2010hadoop} as an example, uploaded file is divided into lots of 128MB segments firstly. 
Then the system allocates those segments to a sequence of blocks one by one, and the size of each block is 128MB as well. 
Finally, the system will generate a file table to record the block addresses and the relation between file segments and blocks. 
The others, like Google file system (GFS) \cite{Ghemawat2003gfs}, have the similar procedure.
For example, the MapReduce programs, which is a popular tool in big data analysis, run on block-based distributed file systems, such as GFS \cite{dean2008mapreduce} and HDFS, etc. 
In the MapReduce programs, the input data is stored in a file partitioned by blocks, and each of the block is processed by a map task.
To reduce the cost of data transmission, the map task is invoked on the data servers where the blocks are stored.
These characteristics lead to the different management in BCSSs from tradition and higher complexity as well. 
As a result, the reliability of BCSSs has been constantly questioned.
For instance, the Amazon Simple Storage Service (S3) accidentally removed a set of servers due to system software errors, resulting in tremendous data loss with many customers involved, which caused great harm to the data integrity in the cloud storage system \cite{S3}.
Therefore the problem of reliability of BCSSs is coming up.

Generally, the reliability of BCSSs is reflected in the correctness and security aspects.
The correctness refers to the algorithm which can produce the expected output for each input \cite{dunlop1980comparative}. 
The security demands a series of mechanisms to protect data from theft or damage. 
Obviously, the correctness is the most basic requirement of reliability.
The correctness analysis not only proves the correctness of the data processing, but also helps to find potential system design errors. 
Therefore, we aim to verify the correctness of BCSSs.
In BCSSs, the correctness refers to the management programs correctness. 
The BCSSs management programs mean a series of commands which are used to deal with the request of user or to manage storage system, such as \textbf{create}, \textbf{append}, and \textbf{delete} commands.
There are plenty of ways to verify the correctness of programs, such as software testing techniques and formal verification.
The former is popular due to its low cost, but the disadvantage of incomplete.
Although the latter is a little more expensive, it can provide strict mathematical analysis support and is efficient in finding errors. Therefore, it is widely used in the safety-critical systems.


Memory management is challenging for formal verification, and memory errors are not easily handled during program execution.
Separation Logic (SL) \cite{reynolds2002separation}, which is a Hoare-style logic, is a well-established approach for formal verification of pointer programs.
SL is best at reasoning about computer memory, especially random access memory \cite{pym2018separation}.
Through introducing a connective $*$ called separating conjunction, SL is able to reason about the shared mutable data structures (e.g., various types of linked lists, binary search trees, and AVL trees).
Besides, it has complete assertions and specification rules support for logical reasoning.
In view of data processing operations of files and blocks in BCSSs management programs are similar to the ordinary computer memory, SL may just meet the requirements.
It seem that SL is a reasonable way to reason about BCSSs, but there are some following challenges that are specific to the complicated target system.
(1). How to model the features of the block-based storage structure? For example, a file consist of a sequence of blocks located by block address. The low-level heap-manipulating framework in SL cannot describe the complicated structure.
(2). How to construct the assertions and specifications according to the block-based storage framework? While the key operations of SL, i.e. separating conjunction $*$, permit the concise and flexible description of shared mutable data structures gracefully, they are not effective in describing the execute actions on the content of blocks.
(3). How to refine the specification rules to make them concise and effective? When considering a algorithm with while-loop, the cumbersome specification is difficult in finding loop invariants.
In order to accommodate to these requirements and restrictions, a novel framework with two-tier heap structure is constructed, several new operators is introduced to describe the properties of BCSSs, and a proof system with Hoare-style specification rules is presented to reason about the block-manipulating programs.

Formal verification is mature enough for developing the correctness of most computer programs.
For example, P. Gardner and G. Nizik \cite{ntzik2015reasoning} proposed their work about local reasoning for the POSIX file system. W.H Hesselink and M.I Lali \cite{hesselink2012formalizing} provided abstract definitions for file systems which are defined as a partial function from paths to data. In addition, some efforts to formalize CSSs have been made. Stephen et al. \cite{james2014program} used formal methods to analyze data flow and proposed an execution model for executing Pig Latin scripts in cloud systems without sacrificing confidentiality of data. I. Pereverzeva et al. \cite{pereverzeva2013formal} founded a formal solution for CSS development, which had the capability of modeling large and elastic data storage system. However, all the above works cannot model or reason about the existed BCSSs.

Before the formal verification work, firstly a formal model of the target system is created, which requires a formal modeling language. However, due to the above special features of BCSSs, it is difficult for existing formal methods to describe the execution process of BCSSs management programs accurately from a view of data details. 
In our previous work \cite{jing2017modeling,wang2018reasoning}, 
we presented a formal language to describe management programs of Massive Data Storage System  (henceforth, LMDSS).
It is an extension of {\bf WHILE}h \cite{yang2001local} programming language with new ingredients to describe file and block operations.
Not only can describe memory management operations, it also describes block-based storage mutation operations.
However, there are three flaws in LMDSS which make it difficult to describe the BCSSs programming model: First, LMDSS regards the content of a block as a big integer.
Although it is sufficient to describe macroscopic file and block operations, for BCSSs management programs the data granularity of LMDSS is not fine enough to accurately describe the process of reading and appending data record in blocks.
Second, the block variables and general variables are disconnected, which means we cannot get a value that can be used for general arithmetic operations from the block content data by some commands.
Third, the assertions and specifications in LMDSS are based on SL, which has a strong theoretical and practical significance in verifying low-level storage model.
But for BCSSs model, these are not enough.
So if we can enhance LMDSS by redefining the block content values, connecting block and general variables, and constructing a proof system, it should be capable to model the BCSSs management programs.

SL \cite{reynolds2002separation} is a mainstream formalization to verify the correctness of traditional storage systems, which has a strong theoretical and practical significance.
By using SL, N.T Smith et al. \cite{birkedal2004local} verified the correctness of Cheney's copying garbage collector programs in memory management system. R. Jung et al. \cite{jung2015iris} extended SL and created the Iris Logic, which supports the verification of concurrent programs. J. Berdine et al. \cite{berdine2005symbolic} proposed the semantics of symbolic heaps, which is a variant of SL. 
In recent years, most of the assertion languages based on SL have adopted the symbolic heaps model. 
Recently, several studies \cite{brotherston2017biabduction,le2017decidable,ta2016automated,brotherston2016model} focused on the logical properties of SL. 
For example, 
Q.T Ta et al. \cite{ta2016automated} presented a sequent-based deductive system for automatically proving entailments by the symbolic heap fragment with arbitrary user-defined inductive heap predicates.
In addition, some verification systems have been implemented in SL  \cite{krogh2017relational,dodds2009deny,svendsen2014impredicative,yang2002semantic,lee2014proof}.

The prior approach of SL to construct verification systems relies on the principle called ``local reasoning'', which was proposed by O'Hearn \cite{o2001local} as a solution for potential pointer aliasing problem in program verification.
The key idea of local reasoning is that even though each phrase of a program can access all stores variables and heaps in principle, it usually only uses a few of them. 
Therefore, a specification and proof can concentrate on only portions of the heaps that a program accesses.
More specifically, two points are worth emphasizing.
Point 1 is that every valid specification $ \{ p \} C \{ q \} $ is ``tight'' in the sense that every cell, which is actually used by $ C $, must be either guaranteed to exist by $ p $ or allocated by $ C $.
Point 2 is a proof method called Frame Rule to formalize local reasoning.
Such rule which is an inference rule, allows one to derive a specification from a given one without ever referring to the actual implementation of a program.
However, SL cannot reason about BCSSs since it is based on a low-level storage model.
Our novel model works directly on the characteristics of BCSSs, which can address all the aforementioned problems.

In this paper, we propose a verification framework to verify the correctness of BCSSs management programs. The main contributions are as follows.

\begin{enumerate}[labelindent=\parindent,leftmargin=*]
	\item A novel framework with two-tier heap structure is constructed to reflect the characteristics of BCSSs, and a modeling language is defined based on the framework. Especially, we introduce the file and block expressions to describe the file-block relationship and refined block content.
	\item Assertions based on SL are constructed to describe the properties of BCSSs.
	Several new operators are introduced to describe the block-heap manipulating operations, e.g., $b_1 \looparrowright b_2$ asserts that the content of $b_1$ is corresponding to $b_2$, which enables the block content to be obtained directly and leaves the address sequence of the block as an implied condition.
	Meanwhile, we introduce quantifiers over block and file variables, which makes the assertions more expressive.
	\item A proof system with the Hoare-style specification rules is proposed to reason about the BCSSs.
	The specification rules, which support local reasoning, are able to describe the behavior of block-manipulating commands concisely, and some special situations are addressed that suffice for formal proofs. An example of verifying a practical algorithm with while-loop is given to demonstrate the feasibility of our method.
\end{enumerate}

The paper is organized as follows. Section \ref{sect2} contains background. Section \ref{sect3} presents a modeling language for BCSSs, and an assertion language based on SL is constructed to describe the properties of BCSSs in Section \ref{sect4}. Section \ref{sect5} proposes the Hoare-style proof system to reason about the BCSSs. Section \ref{sect6} prove that the Frame Rule for BCSSs to reflect local reasoning  is soundness. Section \ref{sect7} gives an illustrating example, and Section \ref{sect8} draws conclusions.

\section{Background} \label{sect2}

In HDFS, the DataNode stores the data blocks into local filesystem directories. However, when writing new blocks to it, there is no guarantee that HDFS will automatically distribute data evenly among the DataNodes in a cluster, so it can easily become imbalanced. That may lead to the frequent use of network bandwidth and low storage efficiency \cite{H1}.

For redistributing data blocks when an imbalance occurs, HDFS provides a command line tool called Disk Balancer.
It is mainly implemented by an algorithm called Transfer, which can let administrators rebalance data across multiple disks by moving blocks from one disk to another [24].
Apparently, moving blocks is the key operation of Disk Balancer, which should not lead to any memory errors such as block losing or content changed.
Hence, formal reasoning is of fundamental importance for this operation.

SL is often used in reasoning about heap-manipulating programs.
However, SL cannot reason about BCSSs since it is difficult to describe the execution process of the management programs accurately from a view of block operation details. 
Taking the Disk Balancer as an example, it involves lots of execute actions on the content of blocks, which is hard for SL to construct the specifications. Besides, SL cannot check whether the file is consistent before or after the execution.   
Therefore, it is necessary to construct a verification framework of BCSSs to prove the correctness problems caused by the complexity of the block-based storage structure according to the characteristics of BCSSs.

The main concepts of the proposed framework is illustrated in the rest of the section.
An initial, high level, and incomplete abstract model of the framework is shown in Fig.\ref{img1}.
To describe the BCSSs management programs, the SL heap is extended by introducing a two-tier heap structure, which is consist of  $\textrm{Heaps}_B$ and $\textrm{Heaps}_V$, while the ordinary store is also extended to reflect the file-block relationship, such as $\textrm{Stores}_F$. The whole computational states of the framework is described later in Sect.\ref{sect3}. In our model, the files are stored as follows. A file variable $f$ is mapped to a sequence of block addresses $bloc_1,...,bloc_k$ by $\textrm{Stores}_F$. $\textrm{Heaps}_B$ maps each of these block addresses to a sequence of location addresses. $\textrm{Heaps}_V$ maps each of these location addresses to a value. In particular, not every location address must belong to a certain block, the rest addresses in $\textrm{Heaps}_V$ are likewise mapped to values. 
Notice that this novel model is not simply a double combination of SL since there is a joint implication between the two tiers of the heap structure, i.e., some location address in $\textrm{Heaps}_V$ belongs to a certain block, which SL cannot express.

\begin{figure}[h]
	\centering
	\includegraphics[scale=0.6]{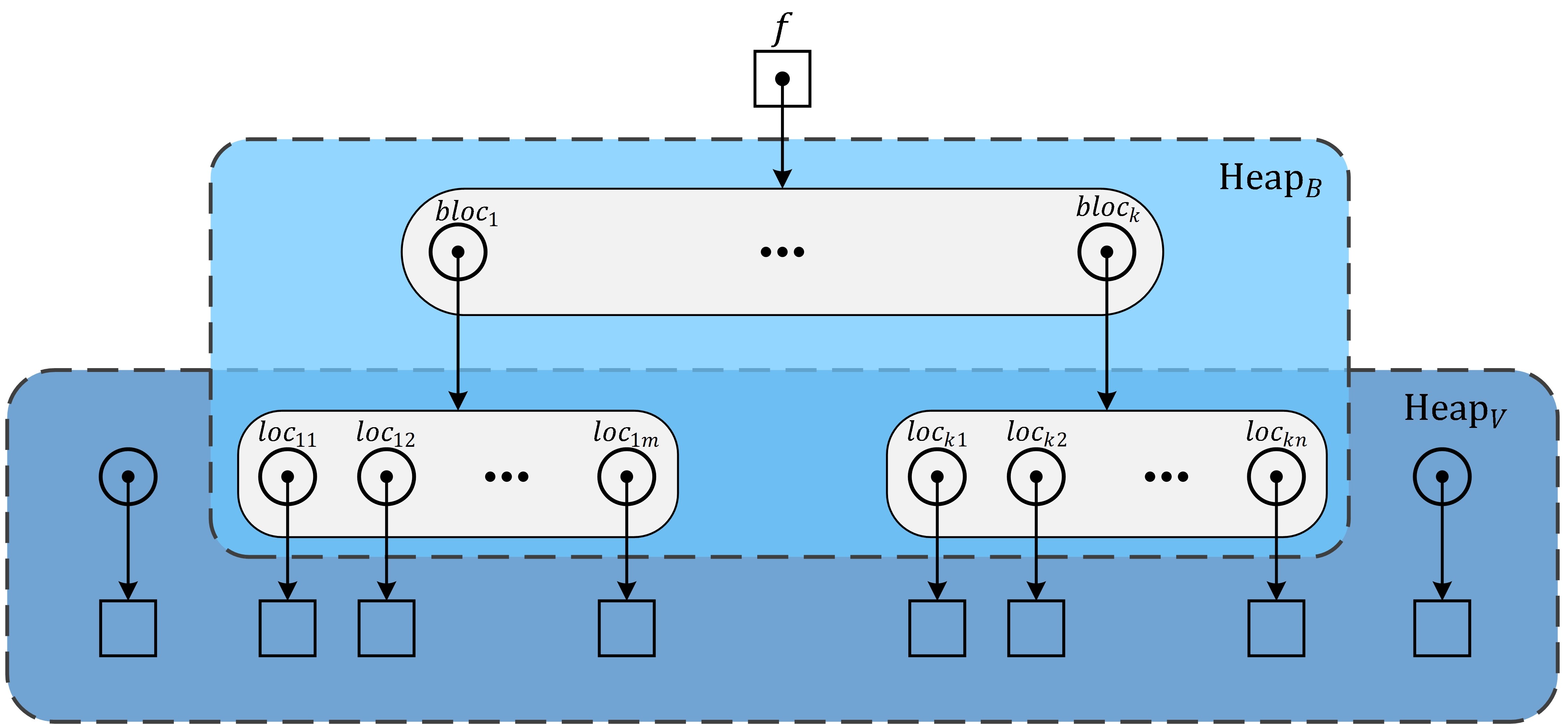}
	\caption{Abstract Model of the Verification Framework}
	\label{img1}
\end{figure}

\section{Modeling Language for BCSSs} \label{sect3}

In this section, the definition of the block content is adjusted based on LMDSS in the introduction in order to have a better description of data processing in BCSSs management programs.
Furthermore, we connect new block variables with general variables to reflect the characteristics of the block-based storage structure.

\subsection{Syntax}
In our language there are four kinds of expressions: (arithmetic) location expressions, file expressions, block expressions, and Boolean expressions. 
The full syntax for expressions and commands in the language is as follows.

\begin{eqnarray*}
	e & ::=  & n ~|~ x,y,... ~|~ e_1 + e_2 ~|~ e_1 - e_2 ~|~ e_1 \times e_2 ~|~ \#f ~|~  \#bk
	\\
	fe & ::= & \mathbf{nil} ~|~ f ~|~ fe \cdot bk ~|~ fe_1 \cdot fe_2
	\\
	bk & ::= & n ~|~ b ~|~ f.e
	\\
	be & ::= & e_1 = e_2 ~|~ e_1 \leq e_2 ~|~ bk_1 == bk_2 ~|~ \mathbf{true} ~|~ \mathbf{false} ~|~ \neg be ~|~ 
	\\
	& & be_1 \land be_2 ~|~ be_1 \lor be_2
	\\
	C & ::= & x:=e ~|~ x:=\mathbf{cons}(\bar{e}) ~|~ x:=[e] ~|~ [e]:=e' ~|~ 
	\\
	& & \mathbf{dispose}(e) ~|~ f:=\mathbf{create}(bk^*) ~|~ \mathbf{attach}(f,bk^*) ~|~ 
	\\
	& &  \mathbf{delete}~f ~|~ b:=\mathbf{allocate}(\bar{e}) ~|~  \mathbf{append}(bk,e) ~|~ 
	\\
	& &  x:=\{bk.e\} ~|~ b:=bk ~|~  \mathbf{delete}~bk ~|~ f.e:=bk  ~|~ 
	\\
	& & C;C' ~|~ \mathbf{if}~be~\mathbf{then}~C~\mathbf{else}~C' ~|~ \mathbf{while}~be~\mathbf{do}~C'
\end{eqnarray*}

where $e$ is written for location expressions, $fe$ for file expressions, $bk$ for block expressions, $be$ for Boolean expressions, and $C$ for commands.

Intuitively, $\#f$ means the block numbers the file $f$ occupies, and $f.e$ points out the address that the $i$-th block of the file $f$ corresponds to, where $i$ is the value of the location expression $e$.
For brevity, we abbreviate the sequence $e_1,...,e_n$ of the location expressions by $\bar{e}$, and the sequence $bk_1,...,bk_2$ of the block expressions by $bk^*$.

Besides the commands in {\bf WHILE}h, which contains all commands of IMP \cite{winskel1993formal}, we introduce some new commands to describe the special operations about files and blocks in the BCSSs management programs.

File commands contain four core operations, which seem to be enough to describe the majority of daily-file operations in CSSs.
\begin{itemize}
	\item $f:=\mathbf{create}(bk^*)$: \hfill file creation;
	\item $\mathbf{attach}(f,bk^*)$: \hfill block address appending;
	\item $\mathbf{delete}~f$:  \hfill file deletion.
\end{itemize}

The file creation command creates a file that consists of block sequence at block addresses expressed by $bk^*$. The block address appending command can append blocks at block addresses $bk^*$ to an existing file.
The file deletion command removes a file and all blocks belong to it become orphaned.

Block commands express block operations.
They are as follows.
\begin{itemize}
	\item $b:=\mathbf{allocate}(\bar{e})$:           \hfill block allocation;
	\item $\mathbf{append}(bk,e)$:               \hfill block content append;
	\item $x:=\{bk.e\}$:         \hfill block content lookup;
	\item $b:=bk$:  \hfill block address assignment;
	\item $f.e:=bk$:  \hfill block address of a file assignment;
	\item $\mathbf{delete}~bk$:  \hfill block deletion.
\end{itemize}

The block allocation command creates a block with the initial value $\bar{e}$, which is the content data of the block. The block content append command appends data $e$ after the content.
Since the content of a block is an integer value sequence at present, the block content lookup command reads the content of the $i$-th field of the block $bk$, where $i$ is the evaluation of expression $e$. 
The block address assignment command assigns value $bk$ to block variable $b$. 
In most of the time, this command is used for the situation involving the file operations, so we add the block address of a file assignment command as an alternative to make the modeling language more powerful.
The block deletion command is formulated similarly as file deletion command with the difference in replacing the file variable by block variable.

%
%
%


\subsection{Domains}

The model has five components:  $\textrm{Stores}_V$,  $\textrm{Stores}_B$, $\textrm{Stores}_F$, $\textrm{Heaps}_V$, and  $\textrm{Heaps}_B$. 
$\textrm{Stores}_V$ is a total function mapping from location variables to addresses. 
$\textrm{Stores}_B$ denotes a total function mapping from block variables to block addresses. Strictly speaking, values, addresses, and block addresses are different in type. 
But in our language, to permit unrestricted address arithmetic, we assume that all the values, addresses, and block addresses are integers. 
$\textrm{Stores}_F$ represents a total function mapping from file variables into a sequence of block addresses. 
$\textrm{Heaps}_V$ is indexed by a subset $ \textrm{Loc} $ of the integers, and it is accessed using indirect addressing $[e]$, where $e$ is a location expression. 
$\textrm{Heaps}_B$ is indexed by a subset $ \textrm{BLoc} $ of the integers, and it is accessed using indirect addressing $\{bk\}$, where $bk$ is a block expression.

\begin{align*}
&\textrm{Values}  \triangleq \{...-1,0,1,...\} = \mathbb{Z}   \quad
\textrm{Loc}, \textrm{BLoc}, \textrm{Atoms} \subseteq \textrm{Values} 
\\
&\textrm{Var}  \triangleq  \{x,y,...\} \quad 
\textrm{BKVar} \triangleq  \{b_1,b_2,...\}  \quad
\textrm{FVar} \triangleq  \{f_1,f_2,...\} \\
&\textrm{Stores}_V  \triangleq \textrm{Var} \rightarrow \textrm{Values}\quad \quad 
\textrm{Stores}_B  \triangleq \textrm{BKVar} \rightarrow \textrm{BLoc} \\
&\textrm{Stores}_F  \triangleq \textrm{FVar} \rightarrow \textrm{BLoc}^* 
\\
&\textrm{Heaps}_B \triangleq \textrm{BLoc} \rightharpoonup_{\textrm{fin}} \textrm{Loc}^* \quad \quad
\textrm{Heaps}_V  \triangleq \textrm{Loc} \rightharpoonup_{\textrm{fin}} \textrm{Values} \quad \quad 
\end{align*}

where $\rightharpoonup$ and $\rightharpoonup_{\textrm{fin}}$ can be found in \cite{hoare1969axiomatic}. $ \textrm{Loc} $, $ \textrm{BLoc} $ and $ \textrm{Atoms} $ are disjoint, and
$\textrm{BLoc}^* = \{ (bloc_1,..., bloc_n)  \mid $ $ bloc_i \in \textrm{BLoc}, n \in \mathbb{N} \}$, and for any element $(bloc_1,..., $ $bloc_n)$ of $\textrm{BLoc}^*$, $|(bloc_1,..., bloc_n)|$ means its length, that is, $|(bloc_1,..., bloc_n)| = n$. Similarly, $\textrm{Loc}^* = \{ (loc_1,..., loc_n)  \mid  loc_i \in \textrm{Loc},  n \in \mathbb{N} \}$.

To make sure the successful allocation, we place the following requirements on the sets of $ \textrm{Loc} $ and $ \textrm{BLoc} $. For any positive integer $m$, there are infinitely many sequences of length $m$ of consecutive integers in $ \textrm{Loc} $. 
For any positive integer $n$, there are infinitely many sequences of length $m$ of discrete integers in $ \textrm{BLoc} $. 
These requirement are satisfied if we take $ \textrm{Loc} $ and $ \textrm{BLoc} $ as the non-negative integers. Then we can take $ \textrm{Atoms} $ as the negative integers, and {\bf nil} as -1.

The states of our language are defined as follows:
\begin{equation*}
\textrm{States} \triangleq \textrm{Stores}_F\times\textrm{Stores}_B\times\textrm{Stores}_V\times\textrm{Heaps}_B\times\textrm{Heaps}_V 
\end{equation*}

A state $\sigma \in \textrm{States}$ is a 5-tuple:  $(s_F,s_B,s_V,h_B,h_V)$.

\subsection{Semantics of the modeling language}

Once the states are defined, we can specify the evaluation rules of our new expressions.
Notice that when we try to give out the semantic of an expression, some stores and heaps may not be used.
For example, expression $\#f$ only needs $\textrm{Stores}_F$ and $\textrm{Stores}_V$.
So we will only list the necessary stores and heaps for each expression.
The semantics of the expressions which is similar to LMDSS can be found in \cite{jing2017modeling}.
Here we give out the semantics of new expressions.
One is the block content length expression $\# bk$, which is used to calculate the content size of the block at address $bk$. Note that in the extended language the content of a block will be a sequence of integers. This expression will count the length of this sequence.
So we can find out how much data(integers) has been written in the current block, and then calculate how much data can be write to this clock, for the maximum size of a block in the system is fixed. The denotational semantics of expression $\# bk$ is in the following.




\begin{scriptsize}
$\begin{array}{l}
\llbracket \# bk \rrbracket (s_F)(s_B)(s_V)(h_B) = k, \\
\text{where} ~ h_B(\llbracket bk \rrbracket (s_F)(s_B)(s_V)(h_B))=(loc_1,loc_2,...,loc_k), \\
\end{array}$
\end{scriptsize}

The others are file expressions, the denotational semantics of file expressions is ruled out by the following functions:


%


\begin{scriptsize}
$\begin{array}{l}
\llbracket \mathbf{nil} \rrbracket (s_F) = ();
\\[1mm]
\llbracket f \rrbracket (s_F) = s_F(f);
\\[1mm]
\llbracket fe \cdot bk \rrbracket (s_F)(s_B) = (bloc_1,\ldots,bloc_n,bloc'), \\
\mbox{if} ~ \llbracket fe \rrbracket (s_F)(s_B)=(bloc_1,\ldots,bloc_n) ~ \mbox{and} ~ \llbracket bk \rrbracket (s_B)=bloc';
\\[1mm]
\llbracket fe_1 \cdot fe_2 \rrbracket (s_F)(s_B) = (bloc_1,\ldots,bloc_n,{bloc_1}',\ldots,{bloc_n}'),\\
\mbox{if} ~ \llbracket fe_1 \rrbracket (s_F)(s_B)=(bloc_1,\ldots,bloc_n) ~ \mbox{and} ~ \\
\llbracket fe_2 \rrbracket (s_F)(s_B)=({bloc_1}',\ldots,{bloc_n}').
\end{array}$
\end{scriptsize}

To state the semantics of the modeling language formally, following \cite{yang2001local}, we use the crucial operations on the heaps:
\begin{itemize}
	\item $\dom(h_H)$ denotes the domain of a heap $h_H \in \textrm{Heaps}_H$, where $h_H$ range over $h_V$ and $h_B$, $\textrm{Heaps}_H$ range over $\textrm{Heaps}_V$ and $\textrm{Heaps}_B$, 
	and $\dom(s_S)$ is the domain of a store $s_S \in \textrm{Stores}_S$, where  $s_S$ range over $s_V$, $s_B$ and $s_F$, $\textrm{Stores}_S$ range over $\textrm{Stores}_V$, $\textrm{Stores}_B$ and $\textrm{Stores}_F$;
	\item $ h_H \# {h_H}' $ indicates that the domains of $h_H$ and ${h_H}'$ are disjoint;
	\item $ h_H * {h_H}' $ is defined when $ h_H \# {h_H}' $ holds and is a finite function obtained by taking the union of $h_H$ and ${h_H}'$;
	\item $ i \mapsto j $ is a singleton partial function which maps $ i $ to $ j $;
	\item for a partial function $ f $ from $ U $ to $ W $ and $ f' $ from $ V $ to $ W $,  the partial function $f[f']$ from $ U $ to $ W $ is defined by:
	$$ f[f'](i) \triangleq \left\{
	\begin{array}{rll}
	&f'(i)               &  \textrm{if} \ i \  \in \dom(f'),\\
	&f(i)                &  \textrm{if} \ i \  \in \dom(f),\\
	&\textrm{undefined}           &\textrm{otherwise}.  \\
	\end{array} \right.  $$
\end{itemize}

Here we give out the denotational semantics of our new commands.

\begin{scriptsize}
	\begin{eqnarray*}
		& & \llbracket f:=\mathbf{create}(bk^*)\rrbracket \sigma = (s_F[(\llbracket bk_1 \rrbracket \sigma,...,\llbracket bk_n \rrbracket \sigma)/f],s_B,s_V,h_B,h_V)
		\\
		& & \mbox{where the term of sequence determined by $h_B(\llbracket bk_i \rrbracket \sigma)$ ($1 \leq i \leq n$) belong }
		\\
		& & \mbox{to} \textrm{dom}(h_V); 
		\\[1mm]
		& & \llbracket \mathbf{attach}(f,bk^*)\rrbracket \sigma =(s_F[(s_F(f)\centerdot (\llbracket bk_1 \rrbracket \sigma,...,\llbracket bk_n \rrbracket \sigma))/f],s_B,s_V,h_B
		\\
		& & ,h_V) \mbox{where the term of sequence determined by $h_B(\llbracket bk_i \rrbracket \sigma)$ ($1 \leq i \leq n$) }
		\\
		& & \mbox{belong to $\textrm{dom}(h_V)$; }
		\\[1mm]
		& & \llbracket \mathbf{delete}~f\rrbracket (s_F,s_B,s_V,h_B,h_V) = (s_F[\llbracket \mathbf{nil}\rrbracket \sigma/f],s_B,s_V,h_B,h_V);\\
		& & \llbracket b:=\mathbf{allocate}(\bar{e}) \rrbracket \sigma = (s_F,s_B[bloc/b],s_V,h_B[(loc_1,...,loc_n)/bloc],
		\\
		& & [h_V|loc_1:\llbracket e_1 \rrbracket \sigma,...,loc_n:\llbracket e_n \rrbracket \sigma])\\
		& & \mbox{where $bloc \in \textrm{BLoc} - \textrm{dom}(h_B)$ , and $loc_1,...,loc_n \in \textrm{Loc} - \textrm{dom}(h_V)$;}
		\\[1mm]
		& & \llbracket \mathbf{append}(bk,e) \rrbracket \sigma = (s_F,s_B,s_V,h_B[(loc_1,...,loc_m,loc_{m+1})/\llbracket bk \rrbracket \sigma],\\
		& & [h_V|loc_{m+1}:\llbracket e \rrbracket \sigma]) ~~~ \mbox{where $ h_B(\llbracket bk \rrbracket \sigma)=(loc_1,...,loc_m) $, the term of }\\
		& & \mbox{ sequence $ (loc_1,...,loc_m) $ belong to $\textrm{dom}(h_V)$, and $loc_{m+1} \in \textrm{Loc} -  $} \\
		& & \mbox{$\textrm{dom}(h_V)$;}
		\\[1mm]
		& & \llbracket x:=\{bk.e\} \rrbracket \sigma = (s_F,s_B,s_V[h_V(loc_i)/x],h_B,h_V)
		\\
		& & \mbox{where $ h_B(\llbracket bk \rrbracket \sigma)=(loc_1,...,loc_m) $,$\llbracket e \rrbracket \sigma = i$ and  $1 \leq i \leq m$;} 
		\\[1mm]
		& & \llbracket b:=bk \rrbracket \sigma = (s_F,s_B[\llbracket bk \rrbracket \sigma/b],s_V,h_B,h_V);
		\\[1mm]
		& & \llbracket f.e:=bk\rrbracket \sigma = (s_F[(bloc_1,...,\llbracket bk \rrbracket \sigma,...,bloc_n)/f],s_B,s_V,h_B,h_V) 	\\
		& & \mbox{where  $s_F(f)=(bloc_1,...,bloc_i,...,bloc_n)$,$\llbracket e \rrbracket \sigma = i$ and  $1 \leq i \leq n$;} 
		\\[1mm]
		& & \llbracket \mathbf{delete}~b \rrbracket \sigma = (s_F,s_B,s_V,h_B\rceil(\textrm{dom}(h_B)-\{\llbracket b \rrbracket \sigma\}),h_V).
	\end{eqnarray*}
\end{scriptsize}

\subsection{A Semantic Example} \label{sect3.4}

At this stage we will give a slightly nontrivial example of a formal proof with our modeling language.
In the example, according to the denotational semantics,
we derive a final state from the initial state by steps.
The final state describes the configuration of the system after running the program,
which achieved the expected results.
This results show the modeling language is expressive enough to describe the management programs in BCSSs, and the semantics is strict and feasible.
The modeling language and its semantic provide the basis for the nature of BCSSs research.
Hence, we can define our assertion language and prove the correctness of management programs.
The description of the example and all proofs are relegated to Appendix A.

\section{Assertion Language for BCSSs} \label{sect4}

To describe the properties of BCSSs, we construct a logic to deal with both locations and blocks. 
\textbf{BI} Pointer Logic \cite{ishtiaq2001bi} provides a powerful formalism for ordinary locations. 
We extend \textbf{BI} pointer logic with the file and block expressions to describe the file-block relationship. Following \cite{yang2001local}, the formal semantics of an assertion is defined by a satisfactory relation `` $ \models $ '' between a state and an assertion. 
$\sigma \models p $ means that the assertion $ p $ holds in the state $ \sigma $.

We rename the assertions of \textbf{BI} pointer logic as \textit{location assertions}, meanwhile call the block formulas as \textit{block assertions}. 
Both location assertions and block assertions are built on expressions.
The key challenge is that the assertion language should be able to express the two-tier structure of the framework.

\subsection{Location Assertion}

\noindent
\textit{Syntax}.
Location assertions describe the properties of locations. 
However, there are some differences between the \textbf{BI} assertions and the location assertions. 
Quantifiers over block and file variables are allowed in location assertions. 
This makes the location assertions much more complicated since they are not first-order quantifiers.

\begin{eqnarray*}
	\alpha & ::= & \mathbf{true}_V ~|~ \mathbf{false}_V ~|~ \neg \alpha ~|~ \alpha_1 \land \alpha_2 ~|~ \alpha_1 \lor \alpha_2 ~|~ \alpha_1 \rightarrow \alpha_2  
	\\
	& &   ~|~ e_1 = e_2 ~|~ e_1 \leq e_2 ~|~  \forall x. \alpha ~|~ \exists x. \alpha ~|~  ~|~ \forall f.\alpha ~|~ \exists f.\alpha ~|~   
	\\
	& & \forall b.\alpha ~|~ \exists b.\alpha ~|~  \mathbf{emp}_V ~|~ e \mapsto e' ~|~ \alpha_1 * \alpha_2 ~|~ \alpha_1 \si \alpha_2
\end{eqnarray*}

\noindent
\textit{Semantics}.
Intuitively, the truth value of a location assertion depends only on the stores and heaps for location variables. 
Given a location assertion $ \alpha $, we define $ \sigma \models \alpha $ by induction on $ \alpha $ in the following.

\noindent
\begin{itemize}
	\item $\sigma \models \mathbf{true}_V$;
	\item $\sigma \models \alpha_1 \rightarrow \alpha_2 ~ \text{iff if} ~ \sigma \models \alpha_1 ~ \text{then} ~ \sigma \models \alpha_2$;
	\item $\sigma \models (\forall x. \alpha) ~ \text{iff} ~ s_F,s_B,s_V[n/x],h_B,h_V \models \alpha \text{ for any } n \in \textrm{Loc}$;
	\item $\sigma \models (\forall f. \alpha) ~ \text{iff} ~
	s_F[\bar{n}/f],s_B,s_V,h_B, h_V \models \alpha $\\
	~~~~~~for any  $\bar{n}=(n_1, n_2, \ldots, n_k) \in \textrm{BlockLoc*}$;
	\item $\sigma \models (\exists f. \alpha) ~ \text{iff} ~ s_F[\bar{n}/f],s_B,s_V,h_B, h_V \models \alpha$\\
	~~~~~~for some $\bar{n}=(n_1, n_2, \ldots, n_k) \in \textrm{BlockLoc*}$;
	\item $\sigma \models \mathbf{emp}_V ~ \text{iff} ~ \text{dom}(h_V) = \emptyset$;
	\item $\sigma \models e \mapsto e' ~ \text{iff} ~ \text{dom}(h_V) = \{\llbracket e \rrbracket \sigma \} \land h_V(\llbracket e \rrbracket \sigma)=\llbracket e' \rrbracket \sigma$;
	\item $\sigma \models \alpha_1 * \alpha_2$ iff there exist $h_V^1, \, h_V^2$ with \\
	~~~~~~$h_V^1 \# h_V^2$ and $h_V=h_V^1 \!*\! h_V^2 $ such that\\
	~~~~~~$s_F,s_B,s_V,h_B,h_V^1 \models \alpha_1$ and $s_F,s_B,s_V,h_B,h_V^2 \models \alpha_2$.
\end{itemize}

Definition of other location assertions are similar to these in \textbf{BI} pointer logic.

\subsection{Block Assertion}

\noindent
\textit{Syntax}.
Block assertions describe the properties about blocks. 
So they include the logic operations and heap operations on block expressions besides ordinary logical connectives and quantifies.

\begin{eqnarray*}
	\beta & ::= & \mathbf{true}_B ~|~ \mathbf{false}_B ~|~ \neg \beta ~|~ \beta_1 \land \beta_2 ~|~ \beta_1 \lor \beta_2 ~|~ \beta_1 \rightarrow \beta_2 
	\\
	& &   ~|~   bk_1 == bk_2 ~|~ bk == bk_1 \circledast  \dots \circledast bk_n ~|~  \forall x. \beta ~|~   
	\\ 
	& & \exists x. \beta ~|~  \forall b. \beta ~|~ \exists b. \beta ~|~ \forall f. \beta ~|~  \exists f. \beta ~|~ \mathbf{emp}_B ~|~  
	\\
	& &   bk \mapsto (\bar{e}) ~|~       \beta_1 * \beta_2 ~|~ \beta_1 \si \beta_2 ~|~ fe=fe' ~|~ b_1 \looparrowright b_2 
\end{eqnarray*}


\noindent
\textit{Semantics}.
Obviously, the truth value of a block assertion depends on the stores and heaps for blocks. 
Given a block assertion $ \beta $, we define $ \sigma \models \beta $ by induction on $ \beta $ as follows.

\noindent

\begin{itemize}
	\item $\sigma \models \mathbf{true}_B$;
	\item $\sigma \models bk_1 == bk_2 ~ \text{iff} ~ \llbracket bk_1 \rrbracket \sigma = \llbracket bk_2 \rrbracket \sigma$;
	\item $\sigma \models bk == bk_1 \circledast \dots \circledast bk_n ~ \text{iff} ~ \llbracket bk \rrbracket \sigma = \llbracket bk_1 \rrbracket \sigma \cdot  \dots \cdot \llbracket bk_n \rrbracket \sigma ~ \text{and} ~  \llbracket bk_i \rrbracket \sigma \bot \llbracket bk_j \rrbracket \sigma$ \\
	~~~~~~ $ ~ \text{for all} ~ i,j \in \mathbb{N}, 1 \leq i < j \leq n$ 
	\item $\sigma \models (\forall x. \beta) ~ \text{iff} ~ s_F,s_B,s_V[n/x],h_B,h_V \models \beta $ \\ 
	~~~~~~ for any  $ n \in \textrm{Loc}$;
	\item $\sigma \models (\exists b. \beta) ~ \text{iff} ~ s_F,s_B[n/b],s_V,h_B,h_V \models \beta $ \\
	~~~~~~ for some $ n \in \textrm{BLoc}$;
	\item $\sigma \models (\forall f. \beta)$ iff $s_F[\bar{n}/f],s_B,s_V,h_B,h_V \models \beta$ \\
	~~~~~~ for any $\bar{n}=(n_1, n_2, \ldots, n_k) \in \textrm{BLoc*}$;
	\item $\sigma \models bk \mapsto (\bar{e})$ iff $\text{dom}(h_B)=\{\llbracket bk \rrbracket \sigma\}   $, \\
	~~~~~~ $ h_B(\llbracket bk \rrbracket \sigma)=(\llbracket e_1 \rrbracket \sigma,...,\llbracket e_n \rrbracket \sigma) $, and \\
	~~~~~~ $ \llbracket e_i \rrbracket \sigma \in \text{dom}(h_V) ~ (1 \leq i \leq n)$;
	\item $\sigma \models (\beta_1 \si \beta_2)$ iff for any block heap ${h_B}'$, \\
	~~~~~~ if ${h_B}' \# h_B$ and $s_F,s_B,s_V,{h_B}',h_V \models \beta_1$, \\
	~~~~~~ then $s_F, s_B,s_V, h_B \!*\! {h_B}',h_V \models \beta_2$;
	\item $\sigma \models fe=fe' ~ \text{iff} ~ \llbracket fe \rrbracket \sigma = \llbracket fe' \rrbracket \sigma$;
	\item $\sigma \models b_1 \looparrowright b_2$ iff $s_B(b_1) \in \mathrm{dom}(h_B)$ when $b_1 = b_2$ or \\
	~~~~~~ $h_V(h_B(s_B(b_1))) = h_V(h_B(s_B(b_2)))$ when $b_1 \not = b_2$.
\end{itemize}

where the notation $\cdot$ means the concatenation of sequences, and $\bot$ denotes two sequences have no common subsequence.

It is noticed that quantifiers over block and file variables are necessary for the assertions. 
In the specification language, we need to express the existence of a sequence of addresses due to the inherent complexity of BCSSs. 
If the length of the sequence is dynamic, for example in a while loop, we cannot use the existence of location variables because the number of the variables is indeterminate.
We address this case by introducing the quantifier over block variable that implicitly shows the dynamic length of the address sequence.

In addition, we define some new block assertions to describe the address sequence and the content of a block, for example, $bk \mapsto (\bar{e})$ and $b_1 \looparrowright b_2$. In these assertions, we straightly show the content of blocks and leaves the address sequence as an implied condition. In the specification language, some commands only change the content of a block, so we do not need to describe the address sequence in pre- and post-conditions, this improvement permits the concise description of the inference rules. Furthermore, with the quantifiers over block variables, the assertion language could express the change of block content flexibly.

\subsection{Global Assertion}

\noindent
\textit{Syntax}.
Roughly speaking, (\textit{location assertion},\textit{block assertion}) pairs are called \textit{global assertions}, 
and the two components of a pair effect each other. 
We use the symbol $ p $ to stand for the global assertions, with the following BNF equation:

\begin{eqnarray*}
	p & ::= & \langle \alpha,\beta \rangle ~|~ \mathbf{true} ~|~ \mathbf{false} ~|~ \neg p ~|~ p_1 \land p_2 ~|~ p_1 \lor p_2 ~|~ p_1 \rightarrow p_2 
	\\
	& &  ~|~    \forall x. p ~|~ \exists x. p ~|~  \forall b. p ~|~ \exists b. p ~|~ \forall f. p ~|~ \exists f. p ~|~  
	\\
	& &  \mathbf{emp} ~|~ p_1 * p_2 ~|~ p_1 \si p_2 
\end{eqnarray*}

\noindent
\textit{Semantics}.
The truth value of a global assertion depends on all kinds of stores and heaps. 
Given a global assertion $ p $, we define $ s_V,s_B,s_F,h_V,h_B \models p $ by induction on $ p $ in the following.

\noindent
\begin{itemize}
	\item $ \mathbf{true} \triangleq  \langle \mathbf{true}_V,\mathbf{true}_B \rangle$; $ \mathbf{emp} \triangleq \langle \mathbf{emp}_V,\mathbf{emp}_B \rangle $;
	\item $\sigma \models \langle \alpha,\beta \rangle$ iff $\sigma \models \alpha$ and $\sigma \models \beta$;
	\item $\sigma \models \neg p$ iff $\sigma \not\models p$;
	\item $\sigma \models p_1 \lor p_2$ iff $\sigma \models p_1$ or $\sigma \models p_2$;
	\item $\sigma \models p_1 \rightarrow p_2$ iff  if $\sigma \models p_1$ then $\sigma \models p_2$;
	\item $\sigma \models \forall x. p$ iff $s_V[n/x],s_B,s_F,h_V,h_B \models p$ \quad for any $n \in \textrm{Loc}$;
	\item $\sigma \models \forall b. p$ iff $s_V,s_B[n/b],s_F,h_V,h_B \models p$ \quad for any $n \in \textrm{BLoc}$;
	\item $\sigma \models \forall f. p$ iff $s_V,s_B,s_F[\bar{n}/f],h_V,h_B \models p$ \\
	~~~~~~ for any $\bar{n}\!=\!(n_1, n_2, \!\ldots\!, n_k) \!\in\! \textrm{BLoc*}$;
	\item $\sigma \models \mathbf{emp}$ iff $\dom(h_V) = \emptyset$ and $\dom(h_B) = \emptyset$;
	\item $s_V,s_B,s_F,h_V,h_B \models p_1 * p_2$ iff there exists $h_H^1, \, h_H^2$ with \\
	~~~~~~ $h_H^1 \# h_H^2$ and $h_H=h_H^1 * h_H^2 $ such that \\
	~~~~~~ $s_V,s_B,s_F,h_V^1,h_B^1 \models p_1$ and $s_V,s_B,s_F,h_V^2,h_B^2 \models p_2$ \\
	where $ h_H^i $ range over $ h_V^i $ and $ h_B^i $ and $ i=1,2 $;
	\item $\sigma \models (p_1 \si p_2)$ iff for any block heap ${h_H}'$, \\
	~~~~~~  if ${h_H}' \# h_H$ and $s_V,s_B,s_F,{h_V}',{h_B}' \models p_1$,\\
	~~~~~~ then $s_V, s_B,s_F, h_V * {h_V}', h_B * {h_B}' \models p_2$, \\
	where $ h_H $ range over $ h_V $ and $ h_B $.
\end{itemize}

Notice that, in the semantic of $\langle \alpha,\beta \rangle$, we do not just simply use $\alpha \land \beta$ here, because such form will mislead the reader that block assertions can be used all alone.
In practical, the content of blocks are stored in $\textrm{Heaps}_V$ for the architecture of BCSSs.
In our setting, location assertions and block assertions work together to describe the blocks correctly, and the pair form can elegantly express the relation between the two tiers as well as the properties of BCSSs.

It is convenient to introduce several complex forms as abbreviations below, which are from \cite{reynolds2002separation}, where some similar terminologies can be found. 
For brevity, we use location variable $l$ to denote the address, and location variable $x$ denotes the content, which will be more readable.

\noindent
\begin{itemize}
	\item $e \mapsto e_1,...,e_n \triangleq e_1 \mapsto e_1 * ... * e+n-1 \mapsto e_n$
	\item $bk \mapsto \bar{l} \triangleq bk \mapsto (l_1,...,l_n)$
	where $\#bk = |\bar{l}| = n$;
	\item $\bar{l} \looparrowright (\bar{e} | i \rightarrowtail x) \triangleq l_1 \mapsto e_1 * ... * l_i \mapsto x * ... * l_n \mapsto e_n$;
	\item $\bar{l} \looparrowright (\bar{e}[x'/x] | i \rightarrowtail x'') \triangleq l_1 \mapsto {e_1}[x'/x] * ... * l_i \mapsto x'' * ... * l_n \mapsto {e_n}[x'/x]$;
	\item $bk \looparrowright \bar{e} \triangleq \exists \bar{x}. \langle \bar{x} \looparrowright \bar{e} , bk \mapsto \bar{x} \rangle $ 
	where $|\bar{x}| = |\bar{e}|$ and the term of sequence $\bar{x}$ are disjoint;
	\item $e \hookrightarrow e'  \triangleq e \mapsto e' * \mathbf{true}_V $;
	\item $bk \hookrightarrow (\bar{e})  \triangleq bk \mapsto (\bar{e}) * \mathbf{true}_B $.
\end{itemize}

%

The advantages of global assertions are threefold:
\begin{enumerate}
	\item The pair form of global assertions is consistent with the hierarchical structural of the framework. Location and block assertions can describe the state of $\textrm{Heaps}_V$ and $\textrm{Heaps}_B$, respectively. 
	Meanwhile, the pair can exactly express the content of each block.
	\item We mainly focus on the properties of blocks in BCSSs. With the new defined notations $==$, $\mapsto$, and $\looparrowright$, we are able to describe the address sequence and the content of a block.
	\item We introduce quantifiers over block and file variables, which makes the assertion language more expressive. Combined with the new defined assertions, we are able to write the pre- and post-conditions in Hoare triples, especially the while loop invariants.
\end{enumerate}

With the advantages, our logic is quite different from SL, and more complicated. The assertion language can support the specification language well. More discussions about assertion language will be discussed in the future.

\section{Hoare-style Proof System for BCSSs} \label{sect5}

For reasoning about management programs in BCSSs, combining with the modeling language and the assertion language above, we introduce a Hoare-style proof system.
The main concept in Hoare-style logic are Hoare triples, which consists of a precondition, a program, and a postcondition.
In our setting, we restrict the pre- and post-conditions to be global assertions only.
Formally, a Hoare triple is of the form $\{ p \} ~C~ \{ q \}$, where $p$ and $q$ are global assertion, and $C$ is a command.

\subsection{Local Reasoning}
To define the semantics of Hoare triples , two point of local reasoning from the Introduction are recalled as follows.

\textbf{Point 1:} Every valid specification $ \{ p \} C \{ q \} $ is ``tight'', that is to say when $ C $ runs in a state satisfying $ p $, it must dereference only those cells that guaranteed to exist by $ p $ or allocated during the execution by possible commands like $ x:=\mathbf{cons}(e_1,...,e_n) $ or $ f:=\mathbf{create}(bk^*) $.

\textbf{Point 2:} An inference rule about specification, known as the Frame Rule, enables us obtain $ \{ p*r \} C \{ q*r \} $ from the initial specification $ \{ p \} C \{ q \} $ of a command when premise of the Frame Rule is valid. 
Thus, we can concentrate on the variables and parts of heaps that are actually accessed by the program which is ensured correct by \textbf{Point 1}, and we can extend a local specification by using the Frame Rule. 
We will proof soundness of the Frame Rule in Sect.\ref{sect6}.

\subsection{Interpretation of Hoare Triples}

\noindent
\textit{\textbf{Interpretation of Hoare Triples}}.
According to \textbf{Point 1}, a tight interpretation of specifications does not dereference the non-addresses, otherwise it will lead to an error \textit{memory fault} or \textit{fault} for short.

We use the terminologies below to specify the certain properties of the program execution.
``$ \langle C,(s_V,s_B,s_F,h_V,h_B) \rangle $ is safe '' when $ \langle C,(s_V,s_B,s_F,h_V,h_B) \rangle $
$\rightsquigarrow^* \textit{fault} $ is impossible;
``$ \langle C,(s_V,s_B,s_F,h_V,h_B) \rangle $ must terminate normally'' when $ \langle C,(s_V,s_B,s_F,h_V,h_B) \rangle $ is safe and 
there exists no infinite $ \rightsquigarrow- $ sequences starting from $ \langle C,(s_V,s_B,s_F,h_V,h_B) \rangle $. Each Hoare triple can be interpreted for partial correctness and for total correctness as follows. 

\textbf{Partial Correctness:}

\begin{tabular}{ll}
	 $ \{ p \} C \{ q \} $ & is $ \mathbf{true} $  iff for $ \forall (s_V,s_B,s_F,h_V,h_B).$\\
	& if $(s_V,s_B,s_F,h_V,h_B) \models p$ then \\
	& 1.$\langle C,(s_V,s_B,s_F,h_V,h_B) \rangle $ is safe  \\
	& 2.  if $ \langle C,(s_V,s_B,s_F,h_V,h_B) \rangle \rightsquigarrow^* $\\
	&  $ ({s_V}',{s_B}',{s_F}',{h_V}',{h_B}')$ \\
	&  then $ ({s_V}',{s_B}',{s_F}',{h_V}',{h_B}') \models q$ 
\end{tabular}

\textbf{Total Correctness:}

\begin{tabular}{ll}
	$ \{ p \} C \{ q \} $ & is $ \mathbf{true} $  iff for $ \forall (s_V,s_B,s_F,h_V,h_B).$\\
	& if $(s_V,s_B,s_F,h_V,h_B) \models p$ then \\
	& 1.$\langle C,(s_V,s_B,s_F,h_V,h_B) \rangle $ must terminate  \\
	& 2.  if $ \langle C,(s_V,s_B,s_F,h_V,h_B) \rangle \rightsquigarrow^* $\\
	&  $({s_V}',{s_B}',{s_F}',{h_V}',{h_B}')$ \\
	&  then $ ({s_V}',{s_B}',{s_F}',{h_V}',{h_B}') \models q$ 
\end{tabular}

\noindent
\textit{\textbf{Transfer Function}}.
According to syntax, any Boolean expression cannot be any kind of assertion. 
Hence, it cannot appear in any specification. 
Thus Boolean expressions in \textbf{if} or \textbf{while} commands cannot be used directly in pre- or post-conditions. 
To fix this problem, we define the transfer function $ T $ that maps Boolean expressions to global assertions.

$T$ $\in$ Transfer functions = Boolean expressions $\rightharpoonup$ Global assertions
\begin{align*}
&
T(e_1 = e_2) = \langle e_1 = e_2, \mathbf{true}_B \rangle
\\
&
T(e_1 \leq e_2) = \langle e_1 \leq e_2, \mathbf{true}_B \rangle
\\
&
T(bk_1 == bk_2) = \langle \mathbf{true}_V, bk_1 == bk_2 \rangle
&
\\
&
T(be_1 \land be_2) = T(be_1) \land T(be_2)
\\
&
T(be_1 \lor be_2) = T(be_1) \lor T(be_2)
\\
&
T(\mathbf{true}) = \mathbf{true} ~~~ T(\mathbf{false}) = \mathbf{false}
~~~
T(\lnot be) = \lnot T(be) 
\\
&   T(\mathbf{abort}) = \mathbf{abort}
\end{align*}


\subsection{The Proof System}
In this section, we present the proof system for BCSSs, which consists of the axioms, the compound command rules, and the structural rules. 
We propose one or more axioms for each basic command. 
Similar to \textbf{BI} pointer logic, the compound command rules and the structural rules does not depend on particular programming constructs.

\noindent
\textit{\textbf{Axioms}}.
Firstly, we give the original axioms of SL. The pre- and post-conditions are both given in the form of $\langle \text{location assertion},\text{block assertion} \rangle$ since they are global assertions.
It is observed that these cases do not involve files or block operations, so all of the rules given here remain valid for our logic.


\begin{scriptsize}
\begin{itemize}
	\item Skip
\end{itemize}
\begin{equation}\label{A1}
\{p\} ~ \mathbf{skip} ~ \{p\} \tag{A1}
\end{equation}

\begin{itemize}
	\item The Simple Assignment form (SA)
\end{itemize}
\begin{equation}\label{A2}
\begin{split}
& \{ \langle x=x' \land \mathbf{emp}_V, \mathbf{emp}_B \rangle \} ~  x:=e  ~  \\
& \{ \langle x=e[x'/x] \land \mathbf{emp}_V ,\mathbf{emp}_B \rangle \}
\end{split} \tag{A2}
\end{equation}
where $x'$ is distinct from $x$.

\begin{itemize}
	\item The Location Allocation form (LA)
\end{itemize}
\begin{equation}\label{A3}
\begin{split}
& \{ \langle x=x' \land \mathbf{emp}_V, \mathbf{emp}_B \rangle \} 
~  x:=\mathbf{cons}(e_1, ... ,e_n)  ~  \\
& \{ \langle x \mapsto e_1[x'/x], ... ,e_n[x'/x] ,\mathbf{emp}_B \rangle \}
\end{split}  \tag{A3}
\end{equation}
where $x'$ is distinct from $x$.

\begin{itemize}
	\item The Location Lookup form (LL)
\end{itemize}
\begin{equation}\label{A4}
\begin{split}
& \{ \langle x=x' \land e \mapsto x'', \mathbf{emp}_B \rangle \} 
~  x:=[e] ~ \\
& \{ \langle x=x'' \land e[x'/x] \mapsto x'' ,\mathbf{emp}_B \rangle \}
\end{split} \tag{A4}
\end{equation}
where $x$,  $x'$, and $x''$ are distinct.

\begin{itemize}
	\item The Location Mutation form (LM)
\end{itemize}
\begin{equation}\label{A5}
\{ \langle e \mapsto - , \mathbf{emp}_B \rangle \} 
~ [e]:=e'  ~
\{ \langle e \mapsto e' ,\mathbf{emp}_B \rangle \} \tag{A5}
\end{equation}

\begin{itemize}
	\item The Deallocation form (DL)
\end{itemize}
\begin{equation}\label{A6}
\{ \langle e \mapsto - , \mathbf{emp}_B \rangle \} 
~ \mathbf{dispose}(e) ~  
\{ \langle \mathbf{emp}_V ,\mathbf{emp}_B \rangle \} \tag{A6}
\end{equation}
\end{scriptsize}

Commands of file will change $\textrm{Stores}_F$, which may have impact on both location assertions and block assertions. To make axioms simple, we only discuss the situation when the file changed in commands does not appear in the location assertion.
In A8, one may argue that it cannot be applied when file expression $\#f$ appears in the assertion.
To this end, we construct an alternative axiom for A8, which is described latter in Sect.\ref{sect5.3}.
For the moment, we ignore particular specifications, and give the common axioms that suffice for formal proofs.

\begin{scriptsize}
\begin{itemize}
	\item The File Creation form (FC) 
\end{itemize}
\begin{equation}\label{A7}
\{ \langle \alpha,f=\mathbf{nil} \land \beta \rangle \} 
~ f:=\mathbf{create}(bk_1, ..., bk_n) ~ \{ \langle \alpha , \beta[(bk_1 , ... , bk_n)/f] \rangle \}\tag{A7}
\end{equation}
where $f$ is not free in $\alpha$ or $ bk_1, ..., bk_n $.

\begin{itemize}
	\item The Block Address Appending form (BAA)
\end{itemize}
\begin{equation}\label{A8}
\begin{split}
&\{ \langle \alpha,f=f' \land \beta \rangle \} 
~ \mathbf{attach}(f,bk_1,...,bk_n) ~ \\
&\{ \langle \alpha[f'/f],\beta[f' \cdot (bk_1[f'/f],...,bk_n[f'/f])/f]  \}
\end{split} \tag{A8}
\end{equation}
where $f'$ is distinct from $f$.

\begin{itemize}
	\item The File Deletion form (FD)
\end{itemize}
\begin{equation}\label{A9}
\{ \langle \alpha , f = f' \land \beta \rangle \} 
~ \mathbf{delete}~f  ~
\{ \langle \alpha[f'/f] , \beta[\mathbf{nil}/f] \rangle \} \tag{A9}
\end{equation}
where $f'$ is distinct from $f$.
\end{scriptsize}

Finally, we come to the block commands for manipulating the $\textrm{Heaps}_B$, which gives rise to a surprising variety of axioms.
For some commands, two axioms are given to make the specification language more expressive.
To explain these axioms, we begin with A10. 
Here, for a precondition with empty $\textrm{Heaps}_V$ and $\textrm{Heaps}_B$ (to show locality), the postcondition says that a new block is created with a sequence of location addresses, and the content of these addresses is $\bar{e}$. 
The restrictions on this axiom are needed to avoid aliasing.
But this axiom cannot be applied generally since the postcondition of a specific form restricts our reasoning. Especially, the quantifiers over a sequence of location addresses cannot appear in while loop invariants.
Therefore, an alternative axiom A11 is given to avoid these complex quantification using block variables.
The rest axioms will be proceed similarly if necessary.
Notice that in A12 and A13, we can not substitute $\mathbf{emp}_V$ for $\alpha$ since $bk$ is well-defined.

When we turn to the axiom for block content lookup, the situation becomes more complicated since such command involves refined block content, i.e., query the content of a certain location address belonging to a block.
The difficulty with this axiom is the accumulation of quantifiers.
In A15, one can think of $b'$ as the first part of block $bk$, and $b''$ denotes the latter part. 
While, $b$ denotes a singleton block heap, the content of which is a single address $l$ with content $x''$.

In most time, we use block address assignment command to assign the content of a file block to a new block, so an axiom A17 is added in the special case of A16 to make the specifications more powerful.
The remaining axioms are relatively well-understood, and not explained in detail.

\begin{scriptsize}

\begin{itemize}
	\item The Block Allocation form (BA) 
\end{itemize}
\begin{equation}\label{A10}
\{ \langle \mathbf{emp}_V, \mathbf{emp}_B \rangle \} 
~ b:=\mathbf{allocate}(\bar{e})   ~  
\{ \exists \bar{l}. \langle \bar{l} \looparrowright \bar{e} , b \mapsto \bar{l} \rangle \} \tag{A10}
\end{equation}
where $b$ is not free in $\bar{e}$.

\begin{itemize}
	\item Alternative Axiom for Block Allocation (BAalt) 
\end{itemize}
\begin{equation}\label{A11}
\{ \langle \mathbf{emp}_V, \mathbf{emp}_B \rangle \} 
~ b:=\mathbf{allocate}(\bar{e})   ~  
\{ \exists b'. \langle \mathbf{true}_V , b == b' \land b'  \looparrowright  b' \rangle \}\tag{A11}
\end{equation}
where $b$ is not free in $\bar{e}$.

\begin{itemize}
	\item The Block Content Append form (BCA) 
\end{itemize}
\begin{equation}\label{A12}
\{ \exists \bar{l'}. \langle \bar{l'} \looparrowright \bar{e'} , bk \mapsto \bar{l'} \rangle \} 
~ \mathbf{append}(bk,e)   ~  
\{ \exists \bar{l'},l. \langle \bar{l'} \looparrowright \bar{e'} * l \mapsto e, bk \mapsto \bar{l'} \cdot (l) \rangle  \}\tag{A12}
\end{equation}
where $bk$ is not free in $\bar{e'},\bar{e}$ and $\#bk$ does not appear in $bk$.

\begin{itemize}
	\item Alternative Axiom for Block Content Append (BCAalt) 
\end{itemize}
\begin{equation}\label{A13}
\begin{split}
& \{ \exists b'. \langle \mathbf{true}_V , bk == b' \land b' \looparrowright b' \rangle \} 
~ \mathbf{append}(bk,e)   ~  \\
& \{ \exists b',b'',l. \langle \mathbf{true}_V * l \mapsto e, bk == b' \circledast b'' \land \mathbf{true}_B * b' \looparrowright b' * b'' \mapsto (l) \rangle  \} 
\end{split} \tag{A13}
\end{equation}
where $bk$ is not free in $\bar{e'},\bar{e}$ and $\#bk$ does not appear in $bk$.

\begin{itemize}
	\item The Block Content Lookup form (BCL) 
\end{itemize}
\begin{equation}\label{A14}
\begin{split}
& \{ \exists \bar{l}. \langle  x=x' \land e=i \land \bar{l} \looparrowright (\bar{e} | i \rightarrowtail x'') , bk \mapsto \bar{l} \rangle \} 
~  x:=\{bk.e\} ~ \\
& \{ \exists \bar{l}. \langle x=x'' \land e[x'/x]=i \land \bar{l} \looparrowright (\bar{e}[x'/x] | i \rightarrowtail x'') , bk[x'/x] \mapsto \bar{l} \rangle \}
\end{split} \tag{A14}
\end{equation}
where $x$,  $x'$, and $x''$ are distinct.

\begin{itemize}
	\item Alternative Axiom for Block Content Lookup (BCLalt) 
\end{itemize}
\begin{equation}\label{A15}
\begin{split}
& \{ \exists b,b',b'',l. \langle  \#b = i-1 \land  x=x' \land e = i \land  l \hookrightarrow x'',
bk == b \circledast b' \circledast b''\\
& \land   b' \hookrightarrow (l)  \rangle \} 
~ x:=\{bk.e\} ~  \{ \exists b,b',b'',l. \langle  \#b = i-1  \land x=x'' \land e[x'/x]  \\
&  = i \land  l \hookrightarrow x'', bk == b \circledast b' \circledast b'' \land   b' \hookrightarrow (l)  \rangle \}
\end{split} \tag{A15}
\end{equation}
where $x$,  $x'$, and $x''$ are distinct.

\begin{itemize}
	\item The Block Address Assignment form (BAA) 
\end{itemize}
\begin{equation}\label{A16}
\{\langle \alpha , b == b' \land \beta \rangle\} ~ b:=bk ~ \{\langle \alpha[b'/b] , b == bk[b'/b] \land \beta[b'/b] \rangle\} \tag{A16}
\end{equation}
where $b'$ is distinct from $b$.

\begin{itemize}
	\item Alternative Axiom for Block Address Assignment (BAAalt) 
\end{itemize}
\begin{equation}\label{A17}
\begin{split}
&\{ \langle \#f_2=i-1 \land \alpha,f = f_2 \cdot b' \cdot f_3  \land \beta\rangle \} 
~ b:=f.i  ~  \\
&\{ \langle \#f_2=i-1 \land \alpha, f = f_2 \cdot b' \cdot f_3 \land b==b'  \land \beta\rangle \} 
\end{split} \tag{A17}
\end{equation}
where $b'$ is distinct from $b$, and $b$ is not free in $\alpha$ or $\beta$.

\begin{itemize}
	\item The Block Address Replacement of a File form (BARF)
\end{itemize}
\begin{equation}\label{A18}
\begin{split}
&\{ \langle  \#f_2=e-1 \land \alpha ,f = f_2 \cdot bk' \cdot f_3 \land \beta \rangle \} 
~ f.e:=bk  ~  \\
&\{ \langle  \#f_2=e-1 \land \alpha,f = f_2 \cdot bk \cdot f_3 \land \beta \rangle \}
\end{split} \tag{A18}
\end{equation}
where $f.e$ does not appear in $\alpha$ or $\beta$.

\begin{itemize}
	\item The Block Deletion form (BD)
\end{itemize}
\begin{equation}\label{A19}
\{ \exists \bar{l}.  \langle \bar{l} \looparrowright - ,b \mapsto \bar{l} \rangle \} 
~ \mathbf{delete}~b  ~
\{ \exists \bar{l}. \langle \bar{l} \looparrowright - ,\mathbf{emp}_B \rangle \} \tag{A19}
\end{equation}
where $b$ does not appear in the expressions which are omitted.
\end{scriptsize}

For space reasons, the soundness of the axioms given above is proved in the appendix of \cite{Jin2019}.

\noindent
\textit{\textbf{Rules}}.	
In contrast to aforementioned axioms, the rules are applicable to arbitrary commands, which are be called structural rules.
In our new setting, the command-specific inference rules and the structural rules of SL remain sound.

Rule of composition applies to sequentially executed programs.
\begin{equation}\label{R1}
\begin{tabular}{c}
$\{p\}~C~\{q\}~~~~\{q\}~C'~\{r\}$
\\
\hline
$\{p \}~C;C'~\{ r \}$
\end{tabular} \tag{R1}
\end{equation}

Conditional rule states that a postcondition common to $\mathbf{then}$ and $\mathbf{else}$ part is also a postcondition of the whole $\mathbf{if}$ statement.

\begin{equation}\label{R2}
\begin{tabular}{c}
$\{p \land T(be) \} C \{ q \}~~~~\{p \land \neg T(be) \} C' \{ q \}$
\\
\hline
$\{p \}~\mathbf{if}~be~\mathbf{then}~C~\mathbf{else}~C'~\{ q \}$
\end{tabular}\tag{R2}
\end{equation}

If the evaluation of $be$ causes an abort, we can use the following rule to show the program aborts.

\begin{equation}\label{R2A}
\{r_1 \land T(be) = \mathbf{abort}\} ~\mathbf{if}~be~\mathbf{then}~c_1~\mathbf{else}~c_2~\{\mathbf{abort}\}\tag{R2A}
\end{equation}

While rule states that the loop invariant is preserved by the loop body.
\begin{equation}\label{R3}
\begin{tabular}{c}
$\{p \land T(be) \} C' \{ p \} $
\\
\hline
$\{p \}~\mathbf{while}~be~\mathbf{do}~C'~ \{ p \land \neg T(be) \}$
\end{tabular}\tag{R3}
\end{equation}

Similar to R2A, R3A shows the program aborts when evaluating $be$.

\begin{equation}\label{R3A}
\{r \land T(be) = \mathbf{abort}\}~\mathbf{while}~be~\mathbf{do}~c~\{\mathbf{abort}\}\tag{R3A}
\end{equation}

We use $ \textrm{FV}(C) $ to denote the set of free variables which occur in $ C $, and $ \textrm{Modify}(C) $ represents the set of variables modified by $ C $, which appear on a left side of an assignment statement. 
Let $ Var_V $ be a set of location variables, $ Var_B $ be a set of block variables, $ Var_F $ be a set of file variables, and $ Var_S $ range over $ Var_V $, $ Var_B $, and $ Var_F $; $ Y_S $ be the set of variables that all allocated variables in stores are precisely in and $ Y_S $ range over $ Y_V $, $ Y_B $, and $ Y_F $; $ X_S $ be the set of variables modified by $ C $ and $ X_S  $ range over  $ X_V $,$ X_B $, and $ X_F $. Then, for commands $ C $ with $ \textrm{Modify}_V (C)=X_V $, $ \textrm{Modify}_B (C)=X_B $, $ \textrm{Modify}_F (C)=X_F $, $ X = X_V \cup X_B \cup X_F $, and $ \textrm{FV}(C) \subseteq Y $, we give the following structural rules.

Consequence Rules:
\begin{equation}\label{R4}
\begin{tabular}{c}
$ \models p'  \rightarrow  p  ~~~~ \{p \} C \{ q \} ~~~~ \models q  \rightarrow  q' $
\\
\hline
$ \{p' \} C \{ q' \}$
\end{tabular}\tag{R4}
\end{equation}

Auxiliary Variable Elimination:

\begin{equation}\label{R5}
\begin{tabular}{c}
$\{p \} C \{ q \} $
\\
\hline
$ \{ \exists x_S.p \} C \{ \exists x_S.q \}$
\end{tabular}
\begin{tabular}{c}
$ x_S \notin \textrm{FV}(C) $
\end{tabular}\tag{R5}
\end{equation}
where $ x_S \in Var_S $ and $ x_S $ range over $ x_V $, $ x_B $, and $ x_F $.

Auxiliary Variable Renaming:
\begin{equation}\label{R6}
\begin{tabular}{c}
$\{p \} C \{ q \} $
\\
\hline
$ \{ p[y_S/x_S] \} C \{ q[y_S/x_S] \}$
\end{tabular}
\begin{tabular}{c}
$ x_S \notin \textrm{FV}(C) \land y_S \notin (\textrm{FV}(C) \cup $ \\ $ \textrm{FV}(p) \cup \textrm{FV}(q) )$
\end{tabular}\tag{R6}
\end{equation}
where $ x_S,y_S \in Var_S $, $ x_S $  range over $ x_V $,$ x_B $，  and $ x_F $, and
$ y_S $ range over $ y_V $, $ y_B $， and $ y_F $.

Frame Rule:
\begin{equation}\label{R7}
\begin{tabular}{c}
$\{p \} C \{ q \} $
\\
\hline
$ \{p*r \} C \{ q*r \}$
\end{tabular}
\begin{tabular}{c}
$ \textrm{Modify}_S (C) \cap \textrm{FV}(r) = \emptyset$
\end{tabular} \tag{R7}
\end{equation}
where $ \textrm{Modify}_S (C) $ range over $ \textrm{Modify}_V (C) $, $ \textrm{Modify}_B (C) $, and $ \textrm{Modify}_F (C) $.

\subsection{More about Specification} \label{sect5.3}

\noindent
\textit{\textbf{Additional Axioms}}.
For several of the axioms we have given, there are particular versions. For instance:

\begin{itemize}
	\item The Block Address Appending form (BAAp)
\end{itemize}
\begin{equation}\label{A20}
\begin{split}
& \{ \langle \#f=m \land \alpha,f=f' \land \beta \rangle \} 
~ \mathbf{attach}(f,bk^*) ~ \\
& \{ \langle \#f=m+n \land \alpha[f'/f],\beta[f' \cdot (bk_1[f'/f],...,bk_n[f'/f])/f]  \}
\end{split} \tag{A20}
\end{equation}
where $f'$ is distinct from $f$.

\begin{itemize}
	\item An instance:
\end{itemize}
\begin{equation}
\begin{split}
& \{ \langle \#f_1=3 \land \#f_2=2,f_1=f_2 \cdot b \rangle \} 
~ \mathbf{attach}(f_1,f_1.3) ~ \\
& \{ \langle \#(f_2 \cdot b)=3 \land \#f_2=2,,f_1=f_2 \cdot b \cdot b  \}
\end{split} 
\end{equation}

In practice, such axiom is rarely used. The only time it is necessary to use is when one must prove a specification with the location expression $\#f$. It can usually be avoided by renaming $\# f$ in the program before proving it. The similar situation will occur in the specifications that include $\#b$, which will be discussed in the future.

\noindent
\textit{\textbf{Scalability Issues}}.
For each of commands, we can give three kinds of inference: local, global, and backward-reasoning. SL has some work in this area \cite{ishtiaq2001bi}, and some similar work for BCSSs is done.
For example, the following instance of block content loop command is the most complex one, which is shown in equation A21. 
Obviously, such backwards version can be applied generally since it works for any postcondition ($p$).

\begin{itemize}
	\item The Block Content Lookup Backward form (BCLBw) 
\end{itemize}
\begin{equation}\label{A21}
\begin{split}
& \{ e=i \land \exists \bar{l},x'.(  \langle   \bar{l} \looparrowright (\bar{e} | i \rightarrowtail x') , bk \mapsto \bar{l} \rangle * (  \langle  \bar{l} \looparrowright (\bar{e}[x'/x] |  \\
& i \rightarrowtail x') , bk[x'/x] \mapsto \bar{l} \rangle \si p[x'/x]  )   )   \}
~  x:=\{bk.e\} ~ \{ p \}
\end{split} \tag{A21}
\end{equation}
where $x'$ is not free in $e$ and $p$.

\section{Soundness of the Frame Rule for BCSSs} \label{sect6}
In this section, we undertake to prove that the Frame Rule for BCSSs is soundness.
The soundness result guarantees that the approach we formalize local reasoning is valid.
Recall the structural rules, we have the syntactic version of the Frame Rule:
\begin{center}
	\begin{tabular}{c}
		$\{p \} - \{ q \} $
		\\
		\hline
		$ \{p*r \} - \{ q*r \}$
	\end{tabular}
	\begin{tabular}{c}
		$ \textrm{Modify}_S (C) \cap \textrm{FV}(r) = \emptyset$
	\end{tabular}
\end{center}
where $ \textrm{Modify}_S (C) $ range over $ \textrm{Modify}_V (C) $, $ \textrm{Modify}_B (C) $, and $ \textrm{Modify}_F (C) $.

Since we treat predicates semantically in this paper, we need to reformulate the Frame Rule to get a semantic version and we need to rewrite the condition in the syntactic version of the Frame Rule.
Let $ \textrm{Modify}_S (C) =X_S $, 
then the condition in the syntactic version of the Frame Rule can be written as:
\begin{equation*}
X_S \cap \textrm{FV}(r) = \emptyset 
\end{equation*}


Now we analyze such condition, review the Sect.\ref{sect4}, the interpretation of an assertion $ r $ is given by $ \{ (s_V,s_B,s_F,h_V,h_B) \mid (s_V,s_B,s_F,h_V,h_B) \in  r \} $.

\begin{enumerate}[labelindent=\parindent,leftmargin=*]
	\item For variables whether location variables, block variables, or file variables, the condition simply checks whether any variable in $ r $ is modified by the command $ C $.
	\item For heap cells which are composed of the heap for variables and heap for blocks, the condition is more elaborate. 
	Consider for the conclusion of the rule, with the definition of ``$ * $'', which says that for every state satisfying $ p*r $ , the current heap can be split into two subheaps so that $ p $ holds for the one and $ r $ for the other. 
	Then for the premise of the rule, the tight interpretation of the Hoare triple $\{p \} C \{ q \} $ says that the command can only access the cells guaranteed to exist by $ p $ which is involved in \textbf{Point1}. 
	So $ r $ is an invariant during the execution which is what the condition says.
\end{enumerate}

Similar to the semantic of $ \forall x.r $, we define:
\begin{eqnarray*}
	&	 \forall X_V.r = \{ (s_V,s_B,s_F,h_V,h_B) \mid \forall {s_V}' \in \textrm{Stores}_V(X_V).\\
	& (s_V[{s_V}'],s_B,s_F,h_V,h_B) \in  r \} ; \\
	&	\forall X_B.r = \{ (s_V,s_B,s_F,h_V,h_B) \mid \forall {s_B}' \in \textrm{Stores}_B(X_B).\\ &(s_V,s_B[{s_B}'],s_F,h_V,h_B) \in  r \} ; \\
	& \forall X_F.r = \{ (s_V,s_B,s_F,h_V,h_B) \mid \forall {s_F}' \in \textrm{Stores}_F(X_F).\\ &(s_V,s_B,s_F[{s_F}'],h_V,h_B) \in  r \}.
\end{eqnarray*}

where $ s_S[{s_S}'] $ denotes the update of $ s_S $ by $ {s_S}' $ defined by:

$$ s_S[{s_S}'](y) \triangleq \left\{
\begin{array}{rll}
&{s_S}'(y)               &  \textrm{if} \ y \  \in \dom({s_S}'),\\
&s_S(y)                  &  \textrm{if} \ y \  \in \dom(s_S),\\
&\textrm{undefined}           &  \textrm{otherwise}.\\
\end{array} \right.  $$

Hence, we have that:
\begin{equation*}
\textrm{Modify}_S (C) \cap \textrm{FV}(r) = \emptyset  \; \; \textrm{holds iff} \; \;  r = \forall X_S. r.
\end{equation*}

The Semantic Version of the Frame Rule:
\begin{center}
	\begin{tabular}{c}
		$\{p \} - \{ q \} $
		\\
		\hline
		$ \{p*r \} - \{ q*r \}$
	\end{tabular}
	\begin{tabular}{c}
		$ r = \forall X_S. r $
	\end{tabular}
\end{center}

where $ X_S $ range over $ X_V $, $ X_B $, and $ X_F $.

\begin{theorem}[Soundness of the Frame Rule]
	The Frame Rule is sound for both partial and total correctness.
\end{theorem}

\begin{proof}
	
	The proof uses the local properties which are from \cite{yang2001local}, where some similar properties are shown.
	
	
	\begin{lemma}[Safety and Termination Monotonicity]
		
		\begin{tabular}{ll}
			$ 1 $ & If $ \langle C,\sigma \rangle $ is safe, and $ h_V \# {h_V}' , h_B \# {h_B}' $, \\
			
			& then $ \langle C,(s_V,s_B,s_F,h_V * {h_V}' ,h_B * {h_B}' ) \rangle $ is safe. \\
			
			$ 2 $  & If $ \langle C,\sigma \rangle $ must terminate, and $ h_V \# {h_V}', h_B \# {h_B}' $, \\
			
			& then $ \langle C,(s_V,s_B,s_F,h_V * {h_V}' ,h_B * {h_B}' ) \rangle $ \\
			& must terminate.
		\end{tabular}	
		
	\end{lemma}
	
	\begin{lemma}[Frame Property]
		
		\begin{tabular}{ll}	
			& Suppose  $ \langle C,(s_V,s_B,s_F,h_V^0,h_B^0) \rangle $ is safe 
			and \\ 
			& \quad $ \langle C,(s_V,s_B,s_F,h_V^0 * h_V^1 ,h_B^0 * h_B^1) \rangle \rightsquigarrow^* $\\
			& $({s_V}',{s_B}',{s_F}',{h_V}',{h_B}')$, then there are $ {h_V^0}' , {h_B^0}' $,\\	
			&  
			 where $ \langle C,(s_V,s_B,s_F,h_V^0,h_B^0) \rangle \rightsquigarrow^* $ $ (s_V,s_B,s_F,{h_V^0}',$  \\
			& 
			${h_B^0}') $ and $ {h_V}' = {h_V^0}' * h_V^1  , {h_B}' = {h_B^0}' * h_B^1 $.
		\end{tabular}
		
	\end{lemma}
	
	\begin{lemma}[Write Locality]  \label{lemma3}
		
		\begin{tabular}{ll}		
			& If $ \langle C,\sigma \rangle \rightsquigarrow^* ({s_V}',{s_B}',{s_F}',{h_V}',{h_B}')$, and\\	
			&  $ x \notin \textrm{Modify}_S (C) $,
			then $ s_S(x) = {s_S}'(x) $.
		\end{tabular}	
		
	\end{lemma}

	We deal with the soundness of the Frame Rule for BCSSs in the following:
	
	(1) For partial correctness:
	
	We will show $ \{p*r \} C \{ q*r \} $ is $ \mathbf{true} $.
	By the interpretation of partial correctness, we have:
	
	\begin{tabular}{ll}
		& $ \{p*r \} C \{ q*r \} $  is  $ \mathbf{true} $  iff for $ \forall (s_V,s_B,s_F,h_V,h_B).$\\
		& if $\sigma \models p*r $ then \\
		& 1.$\langle C,\sigma \rangle $ is safe and \\
		& 2.  if $ \langle C,\sigma \rangle \rightsquigarrow^* ({s_V}',{s_B}',{s_F}',{h_V}',{h_B}')$ \\
		&  then $ ({s_V}',{s_B}',{s_F}',{h_V}',{h_B}') \models q*r $ 
	\end{tabular}
	
	Pick a state $ \sigma $, suppose $\sigma \models p*r $,
	by the definition of  `` $ * $ '', we have:
	
	$\sigma \models p*r $ iff
	
	\hspace*{0.5cm}
	\begin{tabular}[t]{l}
		there exists $h_H^1, \, h_H^2$ with $h_H^1 \# h_H^2$ and $h_H=h_H^1 \!*\! h_H^2 $ \\
		such that $s_V,s_B,s_F,h_V^1,h_B^1 \models p$ and \\
		$s_V,s_B,s_F,h_V^2,h_B^2 \models r$; \\
		where $ h_H^i $ range over $ h_V^i $ and $ h_B^i $ and $ i=1,2 $.
	\end{tabular}
	
	By the premise of the frame rule, we have $ \{p \} C \{ q \} $ is $ \mathbf{true} $,
	
	\begin{tabular}{ll}
		& $ \{p \} C \{ q \} $ is $ \mathbf{true} $  iff for $ \forall (s_V,s_B,s_F,h_V^1,h_B^1).$\\
		& if $(s_V,s_B,s_F,h_V^1,h_B^1) \models p $ then \\
		& 1.$\langle C,(s_V,s_B,s_F,h_V^1,h_B^1) \rangle $ is safe and \\
		& 2.  if $ \langle C,(s_V,s_B,s_F,h_V^1,h_B^1) \rangle \rightsquigarrow^* ({s_V}',{s_B}',{s_F}',$ \\
		&  ${h_V^1}',{h_B^1}')$ then $ ({s_V}',{s_B}',{s_F}',{h_V^1}',{h_B^1}') \models q $ 
	\end{tabular}
	
	By Lemma 1.(1) since $\langle C,(s_V,s_B,s_F,h_V^1,h_B^1)\rangle $ is safe, 
	$\langle C,(s_V,s_B,s_F,$ \\
	$h_V,h_B) \rangle $ is safe, where $h_V=h_V^1 * h_V^2 $ and $h_B=h_B^1 * h_B^2 $.
	
	By Lemma 2, now we have that $\langle C,(s_V,s_B,s_F,h_V^1,h_B^1) \rangle $ is safe, and 
	$ \langle C,(s_V,s_B,s_F,h_V^1 * h_V^2,h_B^1 * h_B^2) \rangle \rightsquigarrow^* ({s_V}',{s_B}',{s_F}',{h_V}',{h_B}')$,
	then there exit $ {h_V^1}' , {h_B^1}' $ such that
	$ \langle C,(s_V,s_B,s_F,h_V^1 ,h_B^1) \rangle \rightsquigarrow^* ({s_V}',{s_B}',{s_F}',{h_V^1}',{h_B^1}')$.
	
	We will show that $ ({s_V}',{s_B}',{s_F}',{h_V}',{h_B}') \models q*r $:
	
	$({s_V}',{s_B}',{s_F}',{h_V}',{h_B}') \models q*r $ iff
	
	\begin{tabular}[t]{l}
		there exist ${h_H^1}', \, h_H^2$ with ${h_H^1}' \# h_H^2$ and $h_H={h_H^1}' * h_H^2 $ \\
		such that $({s_V}',{s_B}',{s_F}',{h_V^1}',{h_B^1}') \models q$ and \\ $({s_V}',{s_B}',{s_F}',h_V^2,h_B^2) \models r$;
	\end{tabular} 
	
	Now we have that $ ({s_V}',{s_B}',{s_F}',{h_V^1}',{h_B^1}') \models q $. By Lemma \ref{lemma3} since $(s_V,s_B,s_F,h_V^2, $ 
	$h_B^2) \models r$ and $ r = \forall X_S. r $,
	we obtain that  $({s_V}',{s_B}',{s_F}',h_V^2,h_B^2) \models r$.
	
	Hence, we find the heaps ${h_H^1}', \, h_H^2$ with ${h_H^1}' \# h_H^2$ and $h_H={h_H^1}' * h_H^2 $ such that
	$ ({s_V}',{s_B}',{s_F}',{h_V^1}',{h_B^1}') \models q $ and $({s_V}',{s_B}',{s_F}',h_V^2,h_B^2) \models r$.
	
	Therefore, $ ({s_V}',{s_B}',{s_F}',{h_V}',{h_B}') \models q*r $,
	Theorem 1 is proved.
	
	(2) The case for total correctness can be handled similarly, the only difference is to use the Termination Monotonicity instead of the Safety Monotonicity. 
	
\end{proof}

\section{An Illustrating Example} \label{sect7}

In this section, we will demonstrate how to use the axioms and rules by a practical example from Disk Balancer (cf.Sect.\ref{sect2}), which proves the correctness of Transfer algorithm as follows.
To make the action of the inference rules clear, each assertion is given by the unabbreviated form.


\begin{algorithm}
	\caption{Transfer}
	\begin{algorithmic}[1]
		\State $b_1 \coloneqq \mathbf{allocate}(1011,1012);$
		\State $f \coloneqq \mathbf{create}(b_1);$	
		\State $b_2 \coloneqq \mathbf{allocate}();$
		\Function {Move}{$~$}
		\State $i \coloneqq 1;$
		\While{$i <= \#b_1$}
		\State $x \coloneqq \{b_1.i\};$	
		\State $\mathbf{append}(b_2,x);$
		\State $i \coloneqq i+1;$
		\EndWhile
		\EndFunction
		\State $f.1 \coloneqq b_2;$
		\State $\mathbf{delete}~b_1;$
	\end{algorithmic}
\end{algorithm}

Here is the full proof.
\begin{scriptsize}
\begin{align*}
& 1 \{ \langle \mathbf{emp}_V, \mathbf{emp}_B \rangle \} &\tag{Given} \\
& \indent b_1 \coloneqq \mathbf{allocate}(1011,1012); \\
&\{ \exists \bar{l}.  \langle \bar{l} \looparrowright (1011,1012), b_1 \mapsto \bar{l} \rangle  \} &\tag{A10} \\
& 2 \{ \exists {b_1}'.  \langle \mathbf{true}_V, b_1 == {b_1}' \land {b_1}' \looparrowright {b_1}' \rangle  \} \tag{1} \\
& 3 \{ \exists {b_1}'.  \langle \mathbf{true}_V, b_1 == {b_1}' \land {b_1}' \looparrowright {b_1}' \rangle  \} \\
& \indent f \coloneqq \mathbf{create}(b_1); \\
& \{ \exists {b_1}'.  \langle \mathbf{true}_V, f = (b_1) \land b_1 == {b_1}' \land {b_1}' \looparrowright {b_1}' \rangle  \} \tag{2, A7} \\
& 4 \{ \exists {b_1}'.  \langle \mathbf{true}_V, f = (b_1) \land b_1 == {b_1}' \land {b_1}' \looparrowright {b_1}' \rangle  \} \\
& \indent b_2 \coloneqq \mathbf{allocate}(); \\
&\{ \exists {b_1}',{b_2}'.  \langle \mathbf{true}_V, f = (b_1) \land b_1 == {b_1}' \land b_2 == {b_2}' \land {b_1}' \looparrowright {b_1}' * {b_2}' \looparrowright \mathbf{nil} \rangle  \} \tag{3, A10}
\end{align*}
\end{scriptsize}

After the executing of while-loop, the content of $b_2$ should be its initial content followed by the content of $b_1$, or formally:
\begin{scriptsize}
\begin{align*}
& \{ \exists {b_1}',{b_2}'. \langle \mathbf{true}_V , b_1 == {b_1}' \land b_2 == {b_2}' \land {b_1}' \looparrowright {b_1}' * {b_2}' \looparrowright \mathbf{nil} \rangle   \}\\
& \indent i \coloneqq 1; \mathbf{while}~be~\mathbf{do}~C';  \\
& \{ \exists {b_1}',{b_2}'. \langle \mathbf{true}_V , b_1 == {b_1}' \land b_2 == {b_2}' \land  \mathbf{true}_B * {b_2}' \looparrowright {b_1}' \rangle \}
\end{align*}
\end{scriptsize}

We are proving the correctness of the Move function step by step. After the first command $i := 1$, a complicated condition is obtained from the initial condition and some unnecessary conditions is omitted in the loop.

\begin{scriptsize}
\begin{align*}
& 5 \{ \exists {b_1}',{b_2}'. \langle 1 \leq \#b_1 + 1 \land \#\mathbf{nil} = 0 \land \mathbf{true}_V , b_1 == (\mathbf{nil},\mathbf{nil}) \circledast {b_1}' \land \\
& b_2 == {b_2}' \circledast (\mathbf{nil},\mathbf{nil}) \land {b_1}' \looparrowright {b_1}' * {b_2}' \looparrowright \mathbf{nil} \rangle   \} \tag{4 } \\
& 6 \{ \exists b_3,b_4,b_5. \langle 1 \leq \#b_1 + 1 \land \#b_3 = 0 \land \#b_4 = \#b_1 \land \mathbf{true}_V , b_1 ==  \\
& b_3 \circledast b_4 \land b_2 == b_5   \land \mathbf{true}_B * b_4 \looparrowright b_4 * b_5 \looparrowright b_3 \rangle \} \tag{5 } \\
& 7 \{ \exists b_3,b_4,b_5. \langle 1 \leq \#b_1 + 1 \land \#b_3 = 0 \land \#b_4 = \#b_1 \land \mathbf{true}_V , b_1 ==  \\
& b_3 \circledast b_4 \land b_2 == b_5   \land \mathbf{true}_B * b_4 \looparrowright b_4 * b_5 \looparrowright b_3 \rangle \} \\
& \indent i \coloneqq 1; \\
& \{ \exists b_3,b_4,b_5. \langle i \leq \#b_1 + 1 \land \#b_3 = i-1 \land \#b_4 = \#b_1 - i + 1 \land \mathbf{true}_V , \\  
& b_1 == b_3 \circledast b_4 \land b_2 == b_5   \land \mathbf{true}_B * b_4 \looparrowright b_4 * b_5 \looparrowright b_3 \rangle \} \tag{6, A2} \\
\end{align*}
\end{scriptsize}

The next command is a while-loop. To use the rule of while-loops (R3), we need to find the loop invariant first.

\begin{scriptsize}
\begin{align*}
\textrm{Let } A & \equiv \exists b_3,b_4,b_5. \langle i \leq \#b_1 + 1 \land \#b_3 = i-1 \land \#b_4 = \#b_1-i+1 \land \\
&  \mathbf{true}_V , b_1 == b_3 \circledast b_4  \land b_2 == b_5 \land \mathbf{true}_B * b_4 \looparrowright b_4 * b_5 \looparrowright b_3 \rangle, \\
be & \equiv i \leq \#b_1, \\
C' & \equiv x \coloneqq \{b_1.i\}; \mathbf{append}(b_2,x); i \coloneqq i+1;
\end{align*}
\end{scriptsize}

Here $b_4$ means the block content copied, $b_5$ means the rest block content. We will prove $A$ is the loop invariant in the following, i.e. $\{A \land T(be)\} ~C'~ \{A\}$.

\begin{scriptsize}
\begin{align*}
& 8 A \land T(be) = \{ \exists b_3,b_4,b_5. \langle i \leq \#b_1  \land \#b_3 = i-1 \land \#b_4 = \#b_1-i+1  \\
&  \land \mathbf{true}_V, b_1 == b_3 \circledast b_4  \land b_2 == b_5 \land \mathbf{true}_B *  b_4 \looparrowright b_4 * b_5 \looparrowright b_3 \rangle \} \\
& 9 \{ \exists b_3,b_5,b_6,b_7,l,y. \langle i \leq \#b_1  \land \#b_3 = i-1 \land \#b_7 = \#b_1-i+1-1 \land  \\
& \mathbf{true}_V * l \mapsto y, b_1 == b_3 \circledast b_6 \circledast b_7 \land b_2 == b_5 \land \mathbf{true}_B * b_6 \mapsto (l) * \\
& b_7 \looparrowright b_7 * b_5 \looparrowright b_3 \rangle \} \tag{8} \\
& 10 \{ \exists b_3,b_5,b_6,b_7,l,y. \langle i \leq \#b_1  \land \#b_3 = i-1 \land \#b_7 = \#b_1-i+1-1   \\
& \land \mathbf{true}_V * l \mapsto y, b_1 == b_3 \circledast b_6 \circledast b_7 \land b_2 == b_5 \land \mathbf{true}_B * b_6 \mapsto (l) *  \\
& b_7 \looparrowright b_7 * b_5 \looparrowright b_3 \rangle \} \\
& \indent x \coloneqq \{b_1.i\}; \\
& \{ \exists b_3,b_5,b_6,b_7,l,y. \langle i \leq \#b_1  \land \#b_3 = i-1 \land \#b_7 = \#b_1-i+1-1 \land  \\
& x=y \land \mathbf{true}_V * l \mapsto y, b_1 == b_3 \circledast b_6 \circledast b_7 \land b_2 == b_5 \land \mathbf{true}_B * b_6 \mapsto   \\
&  (l) * b_7 \looparrowright b_7 * b_5 \looparrowright b_3 \rangle \} \tag{9, A12}\\
& 11 \{ \exists b_3,b_5,b_6,b_7,l,y. \langle i \leq \#b_1  \land \#b_3 = i-1 \land \#b_7 = \#b_1-i+1-1   \\
& \land x=y \land \mathbf{true}_V * l \mapsto y, b_1 == b_3 \circledast b_6 \circledast b_7 \land b_2 == b_5 \land \mathbf{true}_B *  b_6   \\
& \mapsto (l) * b_7 \looparrowright b_7 * b_5 \looparrowright b_3 \rangle \} \\
& \indent \mathbf{append}(b_2,x); \\
& \{\exists b_3,b_5,b_6,b_7,b_8,l,l',y. \langle i \leq \#b_1  \land \#b_3 = i-1 \land \#b_7 = \#b_1-i+1   \\
& -1 \land x=y \land \mathbf{true}_V * l \mapsto y * l' \mapsto y , b_1 == b_3 \circledast b_6 \circledast b_7 \land b_2 == b_5   \\
& \circledast b_8 \land \mathbf{true}_B *  b_6 \mapsto (l) * b_7 \looparrowright b_7 * b_5 \looparrowright b_3 * b_8 \mapsto (l') \rangle \} \tag{10, A11} \\
& 12 \{ \exists b_3,b_6,b_7,b_9,l,l',y. \langle i \leq \#b_1  \land \#b_3 = i-1 \land \#b_7 = \#b_1-i+1 \\
& -1 \land  \mathbf{true}_V * l \mapsto y * l' \mapsto y ,  b_1 == b_3 \circledast b_6 \circledast b_7 \land b_2 == b_9  \land \mathbf{true}_B  \\
&  * b_3 \looparrowright b_3 * b_6 \mapsto (l) * b_7 \looparrowright b_7 * b_9 \mapsto (l') \rangle\} \tag{11} \\
& 13 \{ \exists b_7,b_9,b_{10},l,l',y. \langle i \leq \#b_1  \land \#b_{10} = i-1+1 \land \#b_7 = \#b_1-i+  \\
&  1-1 \land \mathbf{true}_V * l \mapsto y * l' \mapsto y , b_1 == b_{10} \circledast b_7  \land b_2 == b_9  \land \mathbf{true}_B *  \\
& b_{10} \mapsto (l) * b_7 \looparrowright b_7 * b_9 \mapsto (l') \rangle \} \tag{12} \\
& 14 \{ \exists b_7,b_9,b_{10}. \langle i \leq \#b_1  \land \#b_{10} = i-1+1 \land \#b_7 = \#b_1-i+1-1 \land \\
&  \mathbf{true}_V, b_1 == b_{10} \circledast b_7  \land b_2 == b_9  \land \mathbf{true}_B  * b_7 \looparrowright b_7 * b_9 \looparrowright b_{10} \rangle \} \tag{13} \\
& 15 \{ \exists b_7,b_9,b_{10}. \langle i \leq \#b_1  \land \#b_{10} = i-1+1 \land \#b_7 = \#b_1-i+1-1 \\
& \land  \mathbf{true}_V, b_1 == b_{10} \circledast b_7  \land b_2 == b_9  \land \mathbf{true}_B *  b_7 \looparrowright b_7 * b_9 \looparrowright b_{10} \rangle\} \\
& \indent i \coloneqq i+1; \\
& \{ \exists b_7,b_9,b_{10}. \langle i \leq \#b_1 +1  \land \#b_{10} = i-1 \land \#b_7 = \#b_1-i+1 \land  \mathbf{true}_V  ,\\
&   b_1 == b_{10} \circledast b_7  \land b_2 == b_9  \land \mathbf{true}_B *  b_7 \looparrowright b_7 * b_9 \looparrowright b_{10} \rangle\} \tag{14, A2} \\
& 16 \{ \exists b_3,b_4,b_5. \langle i \leq \#b_1 + 1 \land \#b_3 = i-1 \land \#b_4 = \#b_1-i+1 \land \mathbf{true}_V ,\\
&  b_1 == b_3 \circledast b_4  \land b_2 == b_5 \land \mathbf{true}_B *  b_4 \looparrowright b_4 * b_5 \looparrowright b_3 \rangle\} \tag{15}
\end{align*}
\end{scriptsize}

Notice that, $\{ \exists b_3,b_4,b_5. \langle i \leq \#b_1 + 1 \land \#b_3 = i-1 \land \#b_4 = \#b_1-i+1 \land \mathbf{true}_V ,
b_1 == b_3 \circledast b_4  \land b_2 == b_5 \land \mathbf{true}_B *  b_4 \looparrowright b_4 * b_5 \looparrowright b_3 \rangle\}$ in 16 is $A$ itself. So we prove that $\{A \land T(be)\} ~C'~ \{A\}$. By the rule of while-loops (R3), $\{A \} ~ \mathbf{while}~be~\mathbf{do}~C' ~ \{A \land \lnot T(be)\}$ is acquired.

\begin{scriptsize}
\begin{align*}
& 17 A \land \lnot T(be) =  \exists b_3,b_4,b_5. \langle i = \#b_1 + 1 \land \#b_3 = i-1 \land \#b_4 = \#b_1- \\
& i+1 \land \mathbf{true}_V , b_1 == b_3 \circledast b_4  \land b_2 == b_5 \land \mathbf{true}_B *  b_4 \looparrowright b_4 * b_5 \looparrowright b_3 \rangle \\
& 18 \{ \exists b_3,b_4,b_5. \langle i \leq \#b_1 + 1 \land \#b_3 = i-1 \land \#b_4 = \#b_1-i+1 \land \mathbf{true}_V,\\
&  b_1 == b_3 \circledast b_4  \land b_2 == b_5 \land \mathbf{true}_B *  b_4 \looparrowright b_4 * b_5 \looparrowright b_3 \rangle\} \\
& \indent \mathbf{while}~be~\mathbf{do}~C'; \\
&\{ \exists b_3,b_4,b_5. \langle i = \#b_1 + 1 \land \#b_3 = i-1 \land \#b_4 = \#b_1-i+1 \land \mathbf{true}_V ,\\
&  b_1 == b_3 \circledast b_4  \land b_2 == b_5 \land \mathbf{true}_B *  b_4 \looparrowright b_4 * b_5 \looparrowright b_3 \rangle\} \tag{8-16, R3}
\end{align*}
\end{scriptsize}

From the while-loop, we know $i = \#b_1 + 1$ when it finished. The following conclusions is obtained from 18 and the correctness of the Move function is proven.

\begin{scriptsize}
\begin{align*}
& 19 \{ \exists b_3,b_4,b_5. \langle i = \#b_1 + 1 \land \#b_3 = \#b_1 \land \#b_4 = 0 \land \mathbf{true}_V , b_1 == b_3 \\
& \circledast b_4 \land b_2 == b_5 \land \mathbf{true}_B *  b_4 \looparrowright b_4 * b_5 \looparrowright b_3 \rangle\} \tag{18} \\
& 20 \{ \exists b_3,b_5. \langle i = \#b_1 + 1 \land \#b_3 = \#b_1 \land  \mathbf{true}_V, b_1 == b_3  \land b_2 == b_5 \land \\
&  \mathbf{true}_B * b_5 \looparrowright b_3 \rangle\} \tag{19} \\
& 21 \{ \exists {b_1}',{b_2}'. \langle   \mathbf{true}_V , b_1 == {b_1}'  \land b_2 == {b_2}' \land \mathbf{true}_B  * {b_2}' \looparrowright {b_1}' \rangle\} \tag{20} \\
\end{align*}
\end{scriptsize}

After the Move function, using the axioms A15 and A17, the final result can be obtained.

\begin{scriptsize}
\begin{align*}
& 22 \{ \exists {b_1}',{b_2}'. \langle   \mathbf{true}_V , f = (b_1) \land b_1 == {b_1}'  \land  b_2 == {b_2}' \land \mathbf{true}_B  * {b_2}' \looparrowright \\
&  {b_1}' \rangle\} \tag{4, 21} \\
& 23 \{ \exists {b_1}',{b_2}'. \langle   \mathbf{true}_V , f = (b_1) \land b_1 == {b_1}'  \land  b_2 == {b_2}' \land \mathbf{true}_B  * {b_2}' \looparrowright \\
&  {b_1}' \rangle\} \\
& \indent f.1 \coloneqq b_2; \\
& \{ \exists {b_1}',{b_2}'. \langle   \mathbf{true}_V , f = (b_2) \land b_1 == {b_1}'  \land  b_2 == {b_2}' \land \mathbf{true}_B  * {b_2}' \looparrowright {b_1}' \\
&  \rangle\} \tag{22, A15} \\
& 24 \{ \exists \bar{l},\bar{l'}. \langle  \bar{l} \looparrowright (1011,1012) * \bar{l'} \looparrowright (1011,1012),  f = (b_2) \land    b_1 \mapsto \bar{l}  * b_2 \mapsto \\
& \bar{l'} \rangle\} \tag{1, 23} \\
& 25 \{ \exists \bar{l},\bar{l'}. \langle  \bar{l} \looparrowright (1011,1012) * \bar{l'} \looparrowright (1011,1012),  f = (b_2) \land    b_1 \mapsto \bar{l}  * b_2 \mapsto \\
& \bar{l'} \rangle\} \\
& \indent \mathbf{delete}~b_1; \\
&\{ \exists \bar{l},\bar{l'}. \langle  \bar{l} \looparrowright (1011,1012) * \bar{l'} \looparrowright (1011,1012),  f = (b_2) \land  b_2 \mapsto \bar{l'} \rangle\} \tag{24, A17}
\end{align*}
\end{scriptsize}

With all 1-25 and the rule of composition (R1), the correctness of Transfer Algorithm is proven.

\section{Conclusion} \label{sect8}

To insure reliability of block operations in BCSSs and reduce the complexity of the block-based storage
structure and at the same time refine the block content, a verification framework based on SL is introduced to prove the correctness of BCSSs management programs.
The framework is constructed by introducing a novel two-tier heap structure, and a defined modeling language.
Moreover, assertions based on SL is constructed to describe the properties of BCSSs.
Then a proof system with Hoare-style specifications is proposed to reason about the BCSSs.
Using these methods, an example of practical algorithm with while-loop is verified.
The results show that the proving process is scientific and correct.

Future work will focus on the following directions: (1) Consider more characteristics of BCSSs, such as parallelism, (key,value) pairs and
the relations between blocks and locations. (2) Investigate the expressiveness, decidability and model checking algorithms of assertions.
(3) Find out highly efficient proof strategies by selecting adaptive bi-abduction rules to improve the usability of our proof system.


%

\ifCLASSOPTIONcompsoc
  \section*{Acknowledgments}
\else
  \section*{Acknowledgment}
\fi

This work is supported by the National Key R\&D Program of China under Grants 2017YFB1103602 and 2018YFB1003904,
the National Natural Science Foundation of China under Grants 61572003, 61772035, 61751210, 61972005 and 61932001, and the
Major State Research Development Program of China under Grant 2016QY04W0804.

\ifCLASSOPTIONcaptionsoff
  \newpage
\fi

\newpage

\appendices
\section{A Semantic Example}
We will use our modeling language to write a program that tries to create and copy a file. First, we create a file $f_1$ with a sole block $b_1$, which value is $(1011,1012)$. Then we create a new empty block $b_2$, which later is appended with content of file $b_1$ one by one in the while-loop. Finally, we create a new empty file $f_2$, and attach the block $b_2$ to it. When it is finished, we got a file $f_1$ and its copy $f_2$. Note that we cannot change the order of command $\mathbf{L}_2$ and $\mathbf{L}_2$ since that the semantics of $f \coloneqq \mathbf{create}(bk_1,...,bk_n)$  request that the term of sequence determined by $h_B(\llbracket bk_i \rrbracket \sigma)$ belong to $\textrm{dom}(h_V)$.

\begin{algorithm}
	\caption{Copy a File}
	\begin{algorithmic}[1]
		\Function {Copy}{$~$}
		\State $b_1 \coloneqq \mathbf{allocate}(1011,1012);$
		\State $f_1 \coloneqq \mathbf{create}(b_1);$	
		\State $b_2 \coloneqq \mathbf{allocate}();$
		\State $i \coloneqq 1;$
		\While{$i <= \#b_1$}
		\State $x \coloneqq \{b_1.i\};$	
		\State $\mathbf{append}(b_2,x)$
		\State $i \coloneqq i+1;$
		\EndWhile
		\State $f_2 \coloneqq \mathbf{create}();$
		\State $\mathbf{attach}(f_2,b_2);$
		\EndFunction
	\end{algorithmic}
\end{algorithm}

We try to use the denotational semantics to analyze the sample program Algorithm 1.
Assume the initial state is $(s_F,s_B,s_V,h_B,h_V)$, after the execution, we can get a final state.

For convenience, we use label $\mathbf{L}_n$ to express the command in line $n$ of the Algorithm 1,
as an example, $\mathbf{L}_3$ is the command $b_1 \coloneqq \mathbf{allocate}(1011,1012);$.
Also, we use $\mathbf{W}$ to express the while loop part of the program. So there is:

\begin{scriptsize}
\begin{align*}
& ~ \llbracket \mathbf{L_2};\mathbf{L_3};\mathbf{L_4};\mathbf{L_5};\mathbf{W};\mathbf{L_{11}};\mathbf{L_{12}} \rrbracket(s_F,s_B,s_V,h_B,h_V)\\
& = \llbracket \mathbf{L_{12}} \rrbracket(\llbracket \mathbf{L_{11}} \rrbracket(\llbracket \mathbf{W} \rrbracket(\llbracket \mathbf{L_5} \rrbracket(\llbracket \mathbf{L_4} \rrbracket(\llbracket \mathbf{L_3} \rrbracket(\llbracket \mathbf{L_2} \rrbracket(s_F,s_B,s_V,h_B,h_V)))))))\\
& = \llbracket \mathbf{L_{12}} \rrbracket(\llbracket \mathbf{L_{11}} \rrbracket(\llbracket \mathbf{W} \rrbracket(\llbracket \mathbf{L_5} \rrbracket(\llbracket \mathbf{L_4} \rrbracket(\llbracket \mathbf{L_3} \rrbracket(\llbracket b_1 \coloneqq \mathbf{allocate}(1011,1012) \rrbracket(s_F, \\
& s_B,s_V,h_B,h_V)))))))\\
& = \llbracket \mathbf{L_{12}} \rrbracket(\llbracket \mathbf{L_{11}} \rrbracket(\llbracket \mathbf{W} \rrbracket(\llbracket \mathbf{L_5} \rrbracket(\llbracket \mathbf{L_4} \rrbracket(\llbracket \mathbf{L_3} \rrbracket(s_F,s_B[bloc_1/b_1],s_V,h_B[(loc_{11}, \\
& loc_{12})/bloc_1],[h_V|loc_{11}:1011,loc_{12}:1012]))))))\\
& = \llbracket \mathbf{L_{12}} \rrbracket(\llbracket \mathbf{L_{11}} \rrbracket(\llbracket \mathbf{W} \rrbracket(\llbracket \mathbf{L_5} \rrbracket(\llbracket \mathbf{L_4} \rrbracket(\llbracket f_1 \coloneqq \mathbf{create}(b_1) \rrbracket(s_F,s_B[bloc_1/b_1],s_V, \\
& h_B[(loc_{11},loc_{12})/bloc_1],[h_V|loc_{11}:1011,loc_{12}:1012]))))))\\
& = \llbracket \mathbf{L_{12}} \rrbracket(\llbracket \mathbf{L_{11}} \rrbracket(\llbracket \mathbf{W} \rrbracket(\llbracket \mathbf{L_5} \rrbracket(\llbracket \mathbf{L_4} \rrbracket(s_F[(bloc_1)/f_1],s_B[bloc_1/b_1],s_V,h_B \\
& [(loc_{11},loc_{12})/bloc_1],[h_V|loc_{11}:1011,loc_{12}:1012])))))\\
& = \llbracket \mathbf{L_{12}} \rrbracket(\llbracket \mathbf{L_{11}} \rrbracket(\llbracket \mathbf{W} \rrbracket(\llbracket \mathbf{L_5} \rrbracket(\llbracket b_2 \coloneqq \mathbf{allocate}() \rrbracket(s_F[(bloc_1)/f_1],s_B[bloc_1/ \\
& b_1],s_V,h_B[(loc_{11},loc_{12})/bloc_1],[h_V|loc_{11}:1011,loc_{12}:1012])))))\\
& = \llbracket \mathbf{L_{12}} \rrbracket(\llbracket \mathbf{L_{11}} \rrbracket(\llbracket \mathbf{W} \rrbracket(\llbracket \mathbf{L_5} \rrbracket(s_F[(bloc_1)/f_1],s_B[bloc_1/b_1][bloc_2/b_2],s_V, \\
& h_B[(loc_{11},loc_{12})/bloc_1][\mathbf{nil}/bloc_2],[h_V|loc_{11}:1011,loc_{12}:1012]))))\\
& = \llbracket \mathbf{L_{12}} \rrbracket(\llbracket \mathbf{L_{11}} \rrbracket(\llbracket \mathbf{W} \rrbracket(\llbracket i \coloneqq 1 \rrbracket(s_F[(bloc_1)/f_1],s_B[bloc_1/b_1][bloc_2/b_2],s_V, \\
& h_B[(loc_{11},loc_{12})/bloc_1][\mathbf{nil}/bloc_2],[h_V|loc_{11}:1011,loc_{12}:1012]))))\\
& = \llbracket \mathbf{L_{12}} \rrbracket(\llbracket \mathbf{L_{11}} \rrbracket(\llbracket \mathbf{W} \rrbracket(s_F[(bloc_1)/f_1],s_B[bloc_1/b_1][bloc_2/b_2],s_V[1/i],h_B \\
& [(loc_{11},loc_{12})/bloc_1][\mathbf{nil}/bloc_2],[h_V|loc_{11}:1011,loc_{12}:1012])))\\
\end{align*}
\end{scriptsize}

Now we have got an intermediate state after executing commands L2, L3, L4 and L5.
Next, we will analyze the while-loop codes.

For

\begin{scriptsize}
\begin{align*}
& ~ \llbracket i <= \#b_1 \rrbracket(s_F[(bloc_1)/f_1],s_B[bloc_1/b_1][bloc_2/b_2],s_V[1/i],h_B[(loc_{11},\\
& ~ loc_{12})/bloc_1][\mathbf{nil}/bloc_2],[h_V|loc_{11}:1011,loc_{12}:1012])\\
& = \llbracket 1 <= 2 \rrbracket(s_F[(bloc_1)/f_1],s_B[bloc_1/b_1][bloc_2/b_2],s_V[1/i],h_B[(loc_{11},\\
& ~ loc_{12})/bloc_1][\mathbf{nil}/bloc_2],[h_V|loc_{11}:1011,loc_{12}:1012])\\
& = \mathbf{true}
\end{align*}
\end{scriptsize}

Thus

\begin{scriptsize}
\begin{align*}
& ~ \llbracket \mathbf{L_{12}} \rrbracket(\llbracket \mathbf{L_{11}} \rrbracket(\llbracket \mathbf{W} \rrbracket(s_F[(bloc_1)/f_1],s_B[bloc_1/b_1][bloc_2/b_2],s_V[1/i],h_B \\
& ~ [(loc_{11},loc_{12})
/bloc_1][\mathbf{nil}/bloc_2],[h_V|loc_{11}:1011,loc_{12}:1012])))\\
& = \llbracket \mathbf{L_{12}} \rrbracket(\llbracket \mathbf{L_{11}} \rrbracket(\llbracket \mathbf{W} \rrbracket( \llbracket \mathbf{L_9} \rrbracket ( \llbracket \mathbf{L_8} \rrbracket ( \llbracket \mathbf{L_7} \rrbracket (s_F[(bloc_1)/f_1],s_B[bloc_1/b_1] \\
&  [bloc_2/b_2],s_V[1/i],h_B[(loc_{11},loc_{12})
/bloc_1][\mathbf{nil}/bloc_2],[h_V|loc_{11}:1011,\\
& loc_{12}:1012]))))))\\
& = \llbracket \mathbf{L_{12}} \rrbracket(\llbracket \mathbf{L_{11}} \rrbracket(\llbracket \mathbf{W} \rrbracket( \llbracket \mathbf{L_9} \rrbracket ( \llbracket \mathbf{L_8} \rrbracket ( \llbracket x \coloneqq \{b_1.i\} \rrbracket (s_F[(bloc_1)/f_1],s_B[bloc_1 \\
& /b_1][bloc_2/b_2],s_V [1/i],h_B[(loc_{11},loc_{12})
/bloc_1][\mathbf{nil}/bloc_2],[h_V|loc_{11}:\\
& 1011,loc_{12}:1012]))))))\\
& = \llbracket \mathbf{L_{12}} \rrbracket(\llbracket \mathbf{L_{11}} \rrbracket(\llbracket \mathbf{W} \rrbracket( \llbracket \mathbf{L_9} \rrbracket ( \llbracket \mathbf{L_8} \rrbracket  (s_F[(bloc_1)/f_1],s_B[bloc_1/b_1][bloc_2/b_2], \\
& s_V[1/i][1011/x], h_B[(loc_{11},loc_{12})
/bloc_1][\mathbf{nil}/bloc_2],[h_V|loc_{11}:1011,\\
& loc_{12}:1012])))))\\
& = \llbracket \mathbf{L_{12}} \rrbracket(\llbracket \mathbf{L_{11}} \rrbracket(\llbracket \mathbf{W} \rrbracket( \llbracket \mathbf{L_9} \rrbracket ( \llbracket \mathbf{append}(b_2,x) \rrbracket  (s_F[(bloc_1)/f_1],s_B[bloc_1/ \\
& b_1][bloc_2/b_2],s_V [1/i][1011/x],h_B[(loc_{11},loc_{12})
/bloc_1][\mathbf{nil}/bloc_2],[h_V|\\
& loc_{11}:1011,loc_{12}:1012])))))\\
& = \llbracket \mathbf{L_{12}} \rrbracket(\llbracket \mathbf{L_{11}} \rrbracket(\llbracket \mathbf{W} \rrbracket( \llbracket \mathbf{L_9} \rrbracket   (s_F[(bloc_1)/f_1],s_B[bloc_1/b_1][bloc_2/b_2],s_V \\
& [1/i][1011/x],h_B [(loc_{11},loc_{12})
/bloc_1][(loc_{21})/bloc_2],[h_V|loc_{11}:1011,\\
& loc_{12}:1012,loc_{21}:1011]))))\\
& = \llbracket \mathbf{L_{12}} \rrbracket(\llbracket \mathbf{L_{11}} \rrbracket(\llbracket \mathbf{W} \rrbracket( \llbracket i \coloneqq i+1 \rrbracket   (s_F[(bloc_1)/f_1],s_B[bloc_1/b_1][bloc_2/b_2], \\
& s_V[1/i][1011/x], h_B[(loc_{11},loc_{12})
/bloc_1][(loc_{21})/bloc_2],[h_V|loc_{11}:1011,\\
& loc_{12}:1012,loc_{21}:1011]))))\\
& = \llbracket \mathbf{L_{12}} \rrbracket(\llbracket \mathbf{L_{11}} \rrbracket(\llbracket \mathbf{W} \rrbracket (s_F[(bloc_1)/f_1],s_B[bloc_1/b_1][bloc_2/b_2],s_V[2/i][1011 \\
& /x],h_B [(loc_{11},loc_{12})
/bloc_1][(loc_{21})/bloc_2],[h_V|loc_{11}:1011,loc_{12}:1012,\\
& loc_{21}:1011])))\\
\end{align*}
\end{scriptsize}

After the first loop of the while-loop command, block $b_2$ has copied the first content of $b_1$.
And the index variable $i$ equals 2. Therefore, the second loop will be started.

For

\begin{scriptsize}
\begin{align*}
& ~ \llbracket i <= \#b_1 \rrbracket(s_F[(bloc_1)/f_1],s_B[bloc_1/b_1][bloc_2/b_2],s_V[2/i][1011/x],h_B[
\\
& (loc_{11},loc_{12})/bloc_1][(loc_{21})/bloc_2],[h_V|loc_{11}:1011,loc_{12}:1012,loc_{21}:\\
& 1011])\\
& = \llbracket 2 <= 2 \rrbracket(s_F[(bloc_1)/f_1],s_B[bloc_1/b_1][bloc_2/b_2],s_V[2/i][1011/x],h_B[
\\
& (loc_{11},loc_{12})/bloc_1][(loc_{21})/bloc_2],[h_V|loc_{11}:1011,loc_{12}:1012,loc_{21}:\\
& 1011])\\
& = \mathbf{true}
\end{align*}
\end{scriptsize}

Thus

\begin{scriptsize}
\begin{align*}
& ~ \llbracket \mathbf{L_{12}} \rrbracket(\llbracket \mathbf{L_{11}} \rrbracket(\llbracket \mathbf{W} \rrbracket (s_F[(bloc_1)/f_1],s_B[bloc_1/b_1][bloc_2/b_2],s_V[2/i][1011 \\
&  /x],h_B[(loc_{11},loc_{12})
/bloc_1][(loc_{21})/bloc_2],[h_V|loc_{11}:1011,loc_{12}:1012,\\
& loc_{21}:1011])))\\
& = \llbracket \mathbf{L_{12}} \rrbracket(\llbracket \mathbf{L_{11}} \rrbracket(\llbracket \mathbf{W} \rrbracket( \llbracket \mathbf{L_9} \rrbracket ( \llbracket \mathbf{L_8} \rrbracket ( \llbracket \mathbf{L_7} \rrbracket (s_F[(bloc_1)/f_1],s_B[bloc_1/b_1][bloc_2 \\
& /b_2],s_V[2/i][1011/x], h_B[(loc_{11},loc_{12})
/bloc_1][(loc_{21})/bloc_2],[h_V|loc_{11}:\\
& 1011,loc_{12}:1012,loc_{21}:1011]))))))\\
& = \llbracket \mathbf{L_{12}} \rrbracket(\llbracket \mathbf{L_{11}} \rrbracket(\llbracket \mathbf{W} \rrbracket( \llbracket \mathbf{L_9} \rrbracket ( \llbracket \mathbf{L_8} \rrbracket ( \llbracket x \coloneqq \{b_1.i\} \rrbracket (s_F[(bloc_1)/f_1],s_B[bloc_1/ \\
& b_1][bloc_2/b_2],s_V[2/i] [1011/x],h_B[(loc_{11},loc_{12})
/bloc_1][(loc_{21})/bloc_2],[h_V\\
& |loc_{11}:1011,loc_{12}:1012,loc_{21}:1011]))))))\\
& = \llbracket \mathbf{L_{12}} \rrbracket(\llbracket \mathbf{L_{11}} \rrbracket(\llbracket \mathbf{W} \rrbracket( \llbracket \mathbf{L_9} \rrbracket ( \llbracket \mathbf{L_8} \rrbracket (s_F[(bloc_1)/f_1],s_B[bloc_1/b_1][bloc_2/b_2], \\
& s_V[2/i][1012/x],h_B [(loc_{11},loc_{12})
/bloc_1][(loc_{21})/bloc_2],[h_V|loc_{11}:1011,\\
& loc_{12}:1012,loc_{21}:1011])))))\\
\end{align*}
\begin{align*}
& = \llbracket \mathbf{L_{12}} \rrbracket(\llbracket \mathbf{L_{11}} \rrbracket(\llbracket \mathbf{W} \rrbracket( \llbracket \mathbf{L_9} \rrbracket ( \llbracket \mathbf{append}(b_2,x) \rrbracket (s_F[(bloc_1)/f_1],s_B[bloc_1/b_1] \\
& [bloc_2/b_2],s_V[2/i] [1012/x],h_B[(loc_{11},loc_{12})
/bloc_1][(loc_{21})/bloc_2],[h_V|\\
& loc_{11}:1011,loc_{12}:1012,loc_{21}:1011])))))\\
& = \llbracket \mathbf{L_{12}} \rrbracket(\llbracket \mathbf{L_{11}} \rrbracket(\llbracket \mathbf{W} \rrbracket( \llbracket \mathbf{L_9} \rrbracket  (s_F[(bloc_1)/f_1],s_B[bloc_1/b_1][bloc_2/b_2],s_V[2/ \\
& i][1012/x],h_B[(loc_{11}, loc_{12})
/bloc_1][(loc_{21},loc_{22})/bloc_2],[h_V|loc_{11}:1011,\\
& loc_{12}:1012,loc_{21}:1011,loc_{22}:1012]))))\\
& = \llbracket \mathbf{L_{12}} \rrbracket(\llbracket \mathbf{L_{11}} \rrbracket(\llbracket \mathbf{W} \rrbracket( \llbracket i \coloneqq i+1 \rrbracket  (s_F[(bloc_1)/f_1],s_B[bloc_1/b_1][bloc_2/b_2], \\
& s_V[2/i][1012/x],h_B[(loc_{11}, loc_{12})
/bloc_1][(loc_{21},loc_{22})/bloc_2],[h_V|loc_{11}:\\
& 1011,loc_{12}:1012,loc_{21}:1011,loc_{22}:1012]))))\\
& = \llbracket \mathbf{L_{12}} \rrbracket(\llbracket \mathbf{L_{11}} \rrbracket(\llbracket \mathbf{W} \rrbracket (s_F[(bloc_1)/f_1],s_B[bloc_1/b_1][bloc_2/b_2],s_V[3/i][1012 \\
& /x],h_B[(loc_{11},loc_{12})/bloc_1][(loc_{21},(loc_{22})/bloc_2],[h_V|loc_{11}:1011,loc_{12}:\\
& 1012,loc_{21}:1011,loc_{22}:1012])))\\
\end{align*}
\end{scriptsize}

After the second loop, block $b_2$ has copied the other content of $b_1$.Now the index variable $i$ equals 3, which will terminate the while-loop command.

For

\begin{scriptsize}
\begin{align*}
& ~ \llbracket i <= \#b_1 \rrbracket(s_F[(bloc_1)/f_1],s_B[bloc_1/b_1][bloc_2/b_2],s_V[3/i][1012/x],h_B
\\
& [(loc_{11},loc_{12})/bloc_1][(loc_{21},loc_{22})/bloc_2],[h_V|loc_{11}:1011,loc_{12}:1012,\\
& loc_{21}:1011,loc_{22}:1012])\\
& = \llbracket 3 <= 2 \rrbracket(s_F[(bloc_1)/f_1],s_B[bloc_1/b_1][bloc_2/b_2],s_V[3/i][1012/x],h_B
\\
& [(loc_{11},loc_{12})/bloc_1][(loc_{21},loc_{22})/bloc_2],[h_V|loc_{11}:1011,loc_{12}:1012,\\
& loc_{21}:1011,loc_{22}:1012])\\
& = \mathbf{false}
\end{align*}
\end{scriptsize}

Thus

\begin{scriptsize}
\begin{align*}
& ~ \llbracket \mathbf{L_{12}} \rrbracket(\llbracket \mathbf{L_{11}} \rrbracket(\llbracket \mathbf{W} \rrbracket (s_F[(bloc_1)/f_1],s_B[bloc_1/b_1][bloc_2/b_2],s_V[3/i][1012/ \\
&  
x],h_B[(loc_{11},loc_{12})/bloc_1][(loc_{21},loc_{22})/bloc_2],[h_V|loc_{11}:1011,loc_{12}:\\
& 1012,loc_{21}:1011,loc_{22}:1012])))\\
& = \llbracket \mathbf{L_{12}} \rrbracket(\llbracket \mathbf{L_{11}} \rrbracket(s_F[(bloc_1)/f_1],s_B[bloc_1/b_1][bloc_2/b_2],s_V[3/i][1012/x], \\
&  
h_B[(loc_{11},loc_{12})/bloc_1][(loc_{21},loc_{22})/bloc_2],[h_V|loc_{11}:1011,loc_{12}:\\
& 1012,loc_{21}:1011,loc_{22}:1012]))\\
& = \llbracket \mathbf{L_{12}} \rrbracket(\llbracket f_2 \coloneqq \mathbf{create}() \rrbracket(s_F[(bloc_1)/f_1],s_B[bloc_1/b_1][bloc_2/b_2],s_V[3/ \\
&  i][1012/x],h_B[(loc_{11},loc_{12})/bloc_1][(loc_{21},loc_{22})/bloc_2],[h_V|loc_{11}:1011,\\
& loc_{12}:1012,loc_{21}:1011,loc_{22}:1012]))\\
& = \llbracket \mathbf{L_{12}} \rrbracket(s_F[(bloc_1)/f_1][\mathbf{nil}/f_2],s_B[bloc_1/b_1][bloc_2/b_2],s_V[3/i][1012/x], \\
&  h_B[(loc_{11},loc_{12})/bloc_1][(loc_{21},loc_{22})/bloc_2],[h_V|loc_{11}:1011,loc_{12}:\\
& 1012,loc_{21}:1011,loc_{22}:1012])\\
& = \llbracket \mathbf{attach}(f_2,b_2) \rrbracket(s_F[(bloc_1)/f_1][\mathbf{nil}/f_2],s_B[bloc_1/b_1][bloc_2/b_2],s_V[3/ \\
&  i][1012/x],h_B[(loc_{11},loc_{12})/bloc_1][(loc_{21},loc_{22})/bloc_2],[h_V|loc_{11}:1011,\\
& loc_{12}:1012,loc_{21}:1011,loc_{22}:1012])\\
& = (s_F[(bloc_1)/f_1][(bloc_2)/f_2],s_B[bloc_1/b_1][bloc_2/b_2],s_V[3/i][1012/x],\\
& h_B[(loc_{11},loc_{12})/bloc_1]  [(loc_{21},loc_{22})/bloc_2],[h_V|loc_{11}:1011,loc_{12}:\\
& 1012,loc_{21}:1011,loc_{22}:1012])\\
\end{align*}
\end{scriptsize}

Hence

\begin{scriptsize}
\begin{align*}
& ~ \llbracket \mathbf{L_2};\mathbf{L_3};\mathbf{L_4};\mathbf{L_5};\mathbf{W};\mathbf{L_{11}};\mathbf{L_{12}} \rrbracket(s_F,s_B,s_V,h_B,h_V)\\
& = (s_F[(bloc_1)/f_1][(bloc_2)/f_2],s_B[bloc_1/b_1][bloc_2/b_2],s_V[3/i][1012/x],\\
& h_B[(loc_{11},loc_{12})/bloc_1]  [(loc_{21},loc_{22})/bloc_2],[h_V|loc_{11}:1011,loc_{12}:\\
& 1012,loc_{21}:1011,loc_{22}:1012])\\
\end{align*}
\end{scriptsize}

Finally, the program is terminated, and we got a final state.

\section{Soundness of the Specification Axioms}

The soundness of axioms can be proved in denotational semantics by case:

\textbf{Location Commands:}

\begin{itemize}
	\item Skip
\end{itemize}

$\llbracket \mathbf{skip} \rrbracket \sigma = \sigma$\\

\begin{scriptsize}
\begin{tabular}{c}
	$\sigma \models p$
	\\
	\hline
	$\sigma \models p$
\end{tabular}
\end{scriptsize}\\[5mm]

\begin{itemize}
	\item The Simple Assignment form (SA)
\end{itemize}

$\llbracket x:=e \rrbracket \sigma = (s_F,s_B,s_V[\llbracket e \rrbracket \sigma/x],h_B,h_V)$ 
\\

let ${s_V}' \equiv s_V[\llbracket e \rrbracket \sigma/x]$

\begin{scriptsize}
\begin{tabular}{c}
	$\sigma \models \langle x=x' \land \mathbf{emp}_V, \mathbf{emp}_B \rangle$
	\\
	\hline
	$\sigma \models x=x' \land \mathbf{emp}_V$
	$~\text{and}~\sigma \models \mathbf{emp}_B$
	\\
	\hline
	$\llbracket x \rrbracket \sigma=\llbracket x' \rrbracket \sigma, ~ \textrm{dom}(h_V) = \{\},$
	$~\text{and}~ \textrm{dom}(h_B) = \{\}$
	\\
	\hline
	$\textrm{dom}(h_V) = \{\}\land$
	\\
	$ {s_V}'(x) = \llbracket e[x'/x] \rrbracket (s_F,s_B,{s_V}',h_B) ~\text{and}~ \textrm{dom}(h_B) = \{\}$
	\\
	\hline
	$s_F,s_B,{s_V}',h_B,h_V \models x=e[x'/x] \land \mathbf{emp}_V ~\text{and}~ $\\
	$s_F,s_B,{s_V}',h_B,h_V \models \mathbf{emp}_B$
	\\
	\hline
	$s_F,s_B,{s_V}',h_B,h_V \models \langle x=e[x'/x] \land \mathbf{emp}_V , \mathbf{emp}_B \rangle   $
\end{tabular}
\end{scriptsize}
\\[5mm]

\begin{itemize}
	\item The Location Allocation form (LA)
\end{itemize}

$\llbracket x:=\mathbf{cons}(\bar{e})\rrbracket \sigma =(s_F,s_B,s_V[loc/x],h_B,[h_V|loc:\llbracket e_1 \rrbracket \sigma|...|loc+n-1:\llbracket e_n \rrbracket \sigma])$

where $loc,...,loc+n-1 \in \textrm{Loc} - \text{dom}(h_V)$

let ${s_V}'\equiv s_V[loc/x]$ and ${h_V}' \equiv h_V|loc:\llbracket e_1 \rrbracket \sigma|...|loc+n-1:\llbracket e_n \rrbracket \sigma$, then:

\begin{scriptsize}
\begin{tabular}{c}
	$\sigma \models \langle x=x' \land \mathbf{emp}_V, \mathbf{emp}_B \rangle$
	\\
	\hline
	$\sigma \models x=x' \land \mathbf{emp}_V$
	$~\text{and}~\sigma \models \mathbf{emp}_B$
	\\
	\hline
	$\llbracket x \rrbracket \sigma=\llbracket x' \rrbracket \sigma, ~ \textrm{dom}(h_V) = \{\},$
	$~\text{and}~ \textrm{dom}(h_B) = \{\}$
	\\
	\hline
	$\llbracket x \rrbracket \sigma=\llbracket x' \rrbracket \sigma, ~ \textrm{dom}({h_V}') = \{loc,...,loc+n-1\}, ~ {h_V}'(loc) = $
	\\
	$  \llbracket e_1 \rrbracket \sigma \land ... \land {h_V}'(loc+n-1) = \llbracket e_n \rrbracket \sigma, ~\text{and}~ \textrm{dom}(h_B) = \{\}$
	\\
	\hline
	$\textrm{dom}({h_V}') = \{loc,...,loc+n-1\}, ~ {h_V}'(loc) = \llbracket e_1[x'/x] \rrbracket $
	\\
	$  (s_F,s_B,{s_V}',h_B) \land ... \land {h_V}'(loc+n-1) = \llbracket e_n[x'/x] \rrbracket $\\
	$ (s_F,s_B,{s_V}',h_B)~\text{and}~{s_V}'(x) = loc~\text{and}~ \textrm{dom}(h_B) = \{\}$
	\\
	\hline
	$\textrm{dom}({h_V}') = \{loc,...,loc+n-1\}\land {h_V}'(loc) = \llbracket e_1[x'/x] \rrbracket $
	\\
	$ (s_F,s_B,{s_V}',h_B) \land ... \land {h_V}'(loc+n-1) = \llbracket e_n[x'/x] \rrbracket $\\
	$ (s_F,s_B,{s_V}',h_B) ~\text{and}~{s_V}'(x) = loc~\text{and}~ \textrm{dom}(h_B) = \{\}$
	\\
	\hline
	$s_F,s_B,{s_V}',h_B,{h_V}' \models x \mapsto e_1[x'/x],...,e_n[x'/x]~\text{and}~$\\
	$s_F,s_B,{s_V}',h_B,{h_V}' \models \mathbf{emp}_B$
	\\
	\hline
	$s_F,s_B,{s_V}',h_B,{h_V}' \models \langle x \mapsto e_1[x'/x],...,e_n[x'/x], \mathbf{emp}_B \rangle$
\end{tabular}
\end{scriptsize}
\\[5mm]

\begin{itemize}
	\item The Location Lookup form (LL)
\end{itemize}

$\llbracket x:=[e] \rrbracket \sigma = (s_F,s_B,s_V[h_V(\llbracket e \rrbracket \sigma)/x],h_B,h_V)$
\\

let ${s_V}' \equiv s_V[h_V(\llbracket e \rrbracket \sigma)/x]$, then:

\begin{scriptsize}
\begin{tabular}{c}
	$\sigma \models \langle x = x' \land e \mapsto x'', \mathbf{emp}_B \rangle$
	\\
	\hline
	$\sigma \models x = x' \land e \mapsto x''$
	$~\text{and}~ \sigma \models \mathbf{emp}_B$
	\\
	\hline
	$\llbracket x \rrbracket \sigma = \llbracket x' \rrbracket \sigma,$
	$~ \textrm{dom}(h_V) = \{\llbracket e \rrbracket\sigma\}, ~ h_V(\llbracket e \rrbracket\sigma) = \llbracket x'' \rrbracket\sigma, ~\text{and}~$\\
	$ \textrm{dom}(h_B) =\{\}$
	\\
	\hline
	$\textrm{dom}(h_V) = \{\llbracket e[x'/x] \rrbracket(s_F,s_B,{s_V}',h_B)\}, ~$\\
	$h_V(\llbracket e[x'/x] \rrbracket(s_F,s_B,{s_V}',h_B)) = \llbracket x'' \rrbracket\sigma, $
	$~\text{and}~ \textrm{dom}(h_B) =\{\}$
	\\
	\hline
	$\llbracket x \rrbracket (s_F,s_B,{s_V}',h_B,h_V) = h_V(\llbracket e \rrbracket \sigma) = \llbracket x'' \rrbracket (s_F,s_B,{s_V}',h_B,h_V),$
	\\
	$ \textrm{dom}(h_V) = \{\llbracket e[x'/x] \rrbracket(s_F,s_B,{s_V}',h_B,h_V)\}, ~$\\
	$ h_V(\llbracket e[x'/x] \rrbracket(s_F,s_B,{s_V}',h_B,h_V)) = \llbracket x'' \rrbracket(s_F,s_B,{s_V}',h_B,h_V), $
	\\
	$~\text{and}~ \textrm{dom}(h_B) =\{\}$
	\\
	\hline
	$s_F,s_B,{s_V}',h_B,h_V \models x=x'' \land e[x'/x] \mapsto x'' ~\text{and}~$\\
	$s_F,s_B,{s_V}',h_B,h_V \models \mathbf{emp}_B$
	\\
	\hline
	$s_F,s_B,{s_V}',h_B,h_V \models \langle x=x'' \land e[x'/x] \mapsto x'', \beta \rangle$
\end{tabular}
\end{scriptsize}
\\[5mm]

\begin{itemize}
	\item The Location Mutation form (LM)
\end{itemize}

$\llbracket [e]:=e' \rrbracket \sigma = (s_F,s_B,s_V,h_B,h_V[\llbracket e' \rrbracket \sigma/\llbracket e \rrbracket \sigma])$
\\

let ${h_V}' \equiv h_V[\llbracket e' \rrbracket \sigma/\llbracket e \rrbracket \sigma]$, then:

The soundness of axioms can be proved in denotational semantics by case:

\begin{scriptsize}
\begin{tabular}{c}
	$\sigma \models \langle e \mapsto -, \mathbf{emp}_B \rangle$
	\\
	\hline
	$\sigma \models e \mapsto -$
	$~\text{and}~\sigma \models \mathbf{emp}_B$
	\\
	\hline
	$\textrm{dom}(h_V) = \{\llbracket e \rrbracket\sigma\}$
	$~\text{and}~ \textrm{dom}(h_B) =\{\}$
	\\
	\hline
	$\textrm{dom}({h_V}') = \{\llbracket e \rrbracket\sigma\} ~\text{and}~ {h_V}'(\llbracket e \rrbracket\sigma) = \llbracket e' \rrbracket \sigma$
	$~\text{and}~ \textrm{dom}(h_B) =\{\}$
	\\
	\hline
	$s_F,s_B,s_V,h_B,{h_V}' \models e \mapsto e'$
	$~\text{and}~s_F,s_B,s_V,h_B,{h_V}' \models \mathbf{emp}_B$
	\\
	\hline
	$s_F,s_B,s_V,h_B,{h_V}' \models \langle e \mapsto e', \mathbf{emp}_B \rangle$
\end{tabular}
\end{scriptsize}
\\[5mm]

\begin{itemize}
	\item The Deallocation form (DL)
\end{itemize}

$\llbracket \mathbf{dispose}~e \rrbracket \sigma = (s_F,s_B,s_V,h_B,h_V\rceil(\textrm{dom}(h_V)-\{\llbracket e \rrbracket \sigma\}))$
\\

let ${h_V}' \equiv h_V\rceil(\textrm{dom}(h_V)-\{\llbracket e \rrbracket \sigma\}$, then:

\begin{scriptsize}
\begin{tabular}{c}
	$\sigma \models \langle e \mapsto -, \mathbf{emp}_B \rangle$
	\\
	\hline
	$\sigma \models e \mapsto -$
	$~\text{and}~\sigma \models \mathbf{emp}_B$
	\\
	\hline
	$\textrm{dom}(h_V) = \{\llbracket e \rrbracket\sigma\}$
	$~\text{and}~ \textrm{dom}(h_B) =\{\}$
	\\
	\hline
	$\textrm{dom}({h_V}') = \{\llbracket e \rrbracket\sigma\} - \{\llbracket e \rrbracket\sigma\} = \{\}$
	$~\text{and}~ \textrm{dom}(h_B) =\{\}$
	\\
	\hline
	$s_F,s_B,s_V,h_B,{h_V}' \models \mathbf{emp}_V$
	$~\text{and}~s_F,s_B,s_V,h_B,{h_V}' \models \mathbf{emp}_B$
	\\
	\hline
	$s_F,s_B,s_V,h_B,{h_V}' \models \langle \mathbf{emp}_V, \mathbf{emp}_B \rangle$
\end{tabular}
\end{scriptsize}
\\[5mm]

\textbf{File Commands}

\begin{itemize}
	\item The File Creation form (FC)
\end{itemize}

$\llbracket f:=\mathbf{create}(bk^*)\rrbracket \sigma = (s_F[(\llbracket bk_1 \rrbracket \sigma,...,\llbracket bk_n \rrbracket \sigma)/f],s_B,$\\$s_V,h_B,h_V)$

where the term of sequence determined by $h_B(\llbracket bk_i \rrbracket \sigma) ~ (1 \leq i \leq n)$ belong to $\textrm{dom}(h_V)$.

let ${s_F}' \equiv s_F[(\llbracket bk_1 \rrbracket \sigma,...,\llbracket bk_n \rrbracket \sigma)/f]$, as $f$ does not appear in $\alpha$ or $bk_1,...,bk_n$, so we have:

\begin{scriptsize}
\begin{tabular}{c}
	$\sigma \models \langle \alpha,f=\mathbf{nil} \land \beta \rangle$
	\\
	\hline
	$\sigma \models  \alpha, ~ \sigma \models f=\mathbf{nil}, ~\text{and}~ \sigma \models \beta$
	\\
	\hline
	$\sigma \models  \alpha, ~ \llbracket f \rrbracket \sigma = \llbracket \mathbf{nil} \rrbracket \sigma, ~\text{and}~ \sigma \models \beta$
	\\
	\hline
	${s_F}',s_B,s_V,h_B,h_V \models \alpha, ~ {s_F}'(f)=(\llbracket bk_1 \rrbracket ({s_F}',s_B,s_V,h_B),$
	\\
	$...,\llbracket bk_n \rrbracket ({s_F}',s_B,s_V,h_B)) ~\text{and}~ \sigma \models \beta$
	\\
	\hline
	${s_F}',s_B,s_V,h_B,h_V \models \alpha ~\text{and}~$ \\
	${s_F}',s_B,s_V,h_B,h_V \models \beta[(bk_1 , ... , bk_n)/f]$
	\\
	\hline
	${s_F}',s_B,s_V,h_B,h_V \models \langle \alpha , \beta[(bk_1 , ... , bk_n)/f] \rangle$
\end{tabular}
\end{scriptsize}
\\[5mm]

\begin{itemize}
	\item The Block Address Appending form (BAA)
\end{itemize}

$ \llbracket \mathbf{attach}(f,bk^*)\rrbracket \sigma =(s_F[(s_F(f)\centerdot (\llbracket bk_1 \rrbracket \sigma,...,\llbracket bk_n \rrbracket \sigma))$\\$/f],s_B,s_V,h_B,h_V)$

where the term of sequence determined by $h_B(\llbracket bk_i \rrbracket \sigma)$ ($1 \leq i \leq n$) belong to $\textrm{dom}(h_V)$.
\\

let ${s_F}' \equiv s_F[(s_F(f)\centerdot (\llbracket bk_1 \rrbracket \sigma,...,\llbracket bk_n \rrbracket \sigma))/f]$, then:

\begin{scriptsize}
\begin{tabular}{c}
	$\sigma \models \langle \alpha,f=f' \land \beta \rangle$
	\\
	\hline
	$\sigma \models  \alpha, ~ \sigma \models f=f', ~\text{and}~ \sigma \models \beta$
	\\
	\hline
	$\sigma \models  \alpha, ~ \llbracket f \rrbracket \sigma = \llbracket f' \rrbracket \sigma, ~\text{and}~ \sigma \models \beta$
	\\
	\hline
	${s_F}',s_B,s_V,h_B,h_V \models \alpha[f'/f], ~ \sigma \models \beta, ~\text{and}~$
	\\
	${s_F}'(f)=(f' \centerdot (\llbracket bk_1[f'/f] \rrbracket ({s_F}',s_B,s_V,h_B),$\\
	$...,\llbracket bk_n[f'/f] \rrbracket ({s_F}',s_B,s_V,h_B)))$
	\\
	\hline
	${s_F}',s_B,s_V,h_B,h_V \models \alpha[f'/f]~\text{and}~$ \\
	${s_F}',s_B,s_V,h_B,h_V \models \beta[f' \cdot (bk_1[f'/f],...,bk_n[f'/f])/f]$
	\\
	\hline
	${s_F}',s_B,s_V,h_B,h_V \models \langle \alpha[f'/f] , \beta[f' \cdot$\\
	$ (bk_1[f'/f],...,bk_n[f'/f])/f] \rangle$
\end{tabular}
\end{scriptsize}
\\[5mm]

\begin{itemize}
	\item The File Deletion form (FD)
\end{itemize}

$\llbracket \mathbf{delete}~f\rrbracket (s_V,s_B,s_F,h_B,h_V) = (s_F[\llbracket \mathbf{nil} \rrbracket \sigma/f],s_B,s_V,$\\$h_B,h_V)$

let ${s_F}' \equiv s_F[\llbracket \mathbf{nil} \rrbracket \sigma/f]$, then:

\begin{scriptsize}
\begin{tabular}{c}
	$\sigma \models \langle \alpha , f = f' \land \beta \rangle$
	\\
	\hline
	$\sigma \models \alpha, ~ \sigma \models f = f', ~\text{and}~ \sigma \models \beta$
	\\
	\hline
	$\sigma \models  \alpha, ~ \llbracket f \rrbracket \sigma = \llbracket f' \rrbracket \sigma,  ~\text{and}~ \sigma \models \beta$
	\\
	\hline
	${s_F}',s_B,s_V,h_B,h_V \models \alpha[f'/f], ~ {s_F}'(f)=(\llbracket \mathbf{nil} \rrbracket$\\
	$({s_F}',s_B,s_V,h_B)), ~\text{and}~ \sigma \models \beta$
	\\
	\hline
	${s_F}',s_B,s_V,h_B,h_V \models \alpha[f'/f]~\text{and}~$ \\
	${s_F}',s_B,s_V,h_B,h_V \models \beta[\mathbf{nil}/f]$
	\\
	\hline
	${s_F}',s_B,s_V,h_B,h_V \models \langle \alpha[f'/f] , \beta[\mathbf{nil}/f] \rangle$
\end{tabular}
\end{scriptsize}
\\[5mm]

\textbf{Block Commands}

\begin{itemize}
	\item The Block Allocation form (BA)
\end{itemize}

$\llbracket b:=\mathbf{allocate}(\bar{e}) \rrbracket \sigma = (s_F,s_B[bloc/b],s_V,h_B[(loc_1,...,$\\$loc_n)/bloc],[h_V|loc_1:\llbracket e_1 \rrbracket \sigma,...,loc_n:\llbracket e_n \rrbracket \sigma])$

where $bloc \in \textrm{BLoc} - \textrm{dom}(h_B)$ , and $loc_1,...,loc_n \in \textrm{Loc} - \textrm{dom}(h_V)$

let ${s_B}' \equiv s_B[bloc/b]$, ${h_B}'=h_B[(loc_1,...,loc_n)/bloc]$ and ${h_V}'\equiv [h_V|loc_1:\llbracket e_1 \rrbracket \sigma,...,loc_n:\llbracket e_n \rrbracket \sigma]$, where $bloc \in \textrm{BLoc} - \textrm{dom}(h_B)$ , and $loc_1,...,loc_n \in \textrm{Loc} - \textrm{dom}(h_V)$, 
as $b$ is not free in $e_1,...,e_n$, so we have:

\begin{scriptsize}
\begin{tabular}{c}
	$\sigma \models \langle \mathbf{emp}_V , \mathbf{emp}_B \rangle$
	\\
	\hline
	$\sigma \models \mathbf{emp}_V$
	$~\text{and}~\sigma \models \mathbf{emp}_B$
	\\
	\hline
	$\textrm{dom}(h_V) = \{\}$
	$~\text{and}~ \textrm{dom}(h_B) = \{\}$
	\\
	\hline
	$\textrm{dom}({h_V}')=\{loc_1,...,loc_n\} \land {h_V}'(loc_1)=\llbracket e_1 \rrbracket (s_F,{s_B}',[s_V$\\
	$|l_1:loc_1,...,l_n:loc_n],{h_B}') \land ... \land {h_V}'(loc_1)=\llbracket e_n \rrbracket$\\
	$  (s_F,{s_B}',[s_V|l_1:loc_1,...,l_n:loc_n],{h_B}'),  ~  \textrm{dom}({h_B}') = $ \\  
	$ \{bloc\}, ~ {h_B}'(bloc)=(loc_1,...,loc_n), ~\text{and}~ {s_B}'(b)=bloc$ 
	\\
	\hline
	$\textrm{dom}({h_V}')=\{loc_1,...,loc_n\} \land {h_V}'(loc_1)=\llbracket e_1 \rrbracket (s_F,{s_B}',[s_V$\\
	$|l_1:m_1,...,l_n:m_n],{h_B}') \land ... \land {h_V}'(loc_n)=\llbracket e_n \rrbracket (s_F,{s_B}',$\\
	$  [s_V|l_1:m_1,...,l_n:m_n],{h_B}'), ~ $$\textrm{dom}({h_B}') = \{bloc\}, ~ {h_B}'(bloc)=$ \\  
	$  (loc_1,...,loc_n), ~\text{and}~ {s_B}'(b)=bloc, ~\text{for some}~ m_1,...,m_n \in \textrm{Loc}$ 
	\\
	\hline
	$s_F,{s_B}',s_V,{h_B}',{h_V}' \models \exists \bar{l}.\langle  \bar{l} \mapsto \bar{e} , b \mapsto \bar{l}  \rangle$
\end{tabular}
\end{scriptsize}
\\[5mm]

\begin{itemize}
	\item Alternative Axiom for Block Allocation (BAalt) 
\end{itemize}

$\llbracket b:=\mathbf{allocate}(\bar{e}) \rrbracket \sigma = (s_F,s_B[bloc/b],s_V,h_B[(loc_1,...,$\\$loc_n)/bloc],[h_V|loc_1:\llbracket e_1 \rrbracket \sigma,...,loc_n:\llbracket e_n \rrbracket \sigma])$

where $bloc \in \textrm{BLoc} - \textrm{dom}(h_B)$ , and $loc_1,...,loc_n \in \textrm{Loc} - \textrm{dom}(h_V)$

let ${s_B}' \equiv s_B[bloc/b]$, ${s_B}'' \equiv {s_B}'[bloc/b']$, ${h_B}'=h_B[(loc_1,$\\$...,loc_n)/bloc]$ and ${h_V}'\equiv [h_V|loc_1:\llbracket e_1 \rrbracket \sigma,...,loc_n:\llbracket e_n \rrbracket \sigma]$, where $bloc \in \textrm{BLoc} - \textrm{dom}(h_B)$ , and $loc_1,...,loc_n \in \textrm{Loc} - \textrm{dom}(h_V)$, 
as $b$ is not free in $e_1,...,e_n$, so we have:

\begin{scriptsize}
\begin{tabular}{c}
	$\sigma \models \langle \mathbf{emp}_V , \mathbf{emp}_B \rangle$
	\\
	\hline
	$\sigma \models \mathbf{emp}_V$
	$~\text{and}~\sigma \models \mathbf{emp}_B$
	\\
	\hline
	$\textrm{dom}(h_V) = \{\}$
	$~\text{and}~\textrm{dom}(h_B) = \{\}$
	\\
	\hline
	${s_B}' = bloc ~\text{and}~  \textrm{dom}(h_B) = \{bloc\} $
	\\
	\hline
	${s_B}''(b)={s_B}'[bloc/b'](b)=s_B[bloc/b][bloc/b'](b)=bloc, ~$ \\
	${s_B}''(b')={s_B}'[bloc/b'](b')=bloc, ~\text{and}~ \textrm{dom}(h_B) = \{bloc\}$
	\\
	\hline
	$s_F,{s_B}'[bloc/b'],s_V,{h_B}',{h_V}' \models b==b' \land b'  \looparrowright  b' $ \\
	$s_F,{s_B}'[bloc/b'],s_V,{h_B}',{h_V}' \models \mathbf{true}_V$
	\\
	\hline
	$s_F,{s_B}'[m/b'],s_V,{h_B}',{h_V}' \models b==b' \land b'  \looparrowright  b' $ \\
	$s_F,{s_B}'[m/b'],s_V,{h_B}',{h_V}' \models \mathbf{true}_V$ for some $m \in \textrm{BLoc}$
	\\
	\hline
	$s_F,{s_B}'[m/b'],s_V,{h_B}',{h_V}' \models \langle \mathbf{true}_V , b == b' \land b'  \looparrowright  b' \rangle$ \\
	for some $m \in \textrm{BLoc}$
	\\
	\hline
	$s_F,{s_B}',s_V,{h_B}',{h_V}' \models \exists b'. \langle \mathbf{true}_V , b == b' \land b'  \looparrowright  b' \rangle$
\end{tabular}
\end{scriptsize}
\\[5mm]

\begin{itemize}
	\item The Block Content Append form (BCA) 
\end{itemize}

$\llbracket \mathbf{append}(bk,e) \rrbracket \sigma = (s_F,s_B,s_V,h_B[(loc_1,...,loc_m,$\\
$loc_{m+1})/\llbracket bk \rrbracket \sigma],[h_V|loc_{m+1}:\llbracket e \rrbracket \sigma])$

where $ h_B(\llbracket bk \rrbracket \sigma)=(loc_1,...,loc_m) $, the term of sequence 
$ (loc_1,...,loc_m) $ belong to $\textrm{dom}(h_V)$, and $loc_{m+1} \in \textrm{Loc} - \textrm{dom}(h_V)$;

let ${h_B}' \equiv h_B[(loc_1,...,loc_m,loc_{m+1})/\llbracket bk \rrbracket \sigma]$ and ${h_V}' \equiv [h_V|loc_{m+1}:\llbracket e \rrbracket \sigma]$, where $h_B(\llbracket bk \rrbracket \sigma)=(loc_1,...,loc_m)$ , and $loc_{m+1} \in \textrm{Loc} - \textrm{dom}(h_V)$, as $bk$ is not free in ${e_1}',...,{e_m}',e$ and $\#bk$ does not appear in $bk$. Note that we can get $\llbracket bk \rrbracket (s_F,s_B,s_V,{h_B}')$ from $\llbracket bk \rrbracket \sigma$, so we have:

\begin{scriptsize}
\begin{tabular}{c}
	$\sigma \models \exists \bar{l'}. \langle \bar{l'} \looparrowright \bar{e'} , bk \mapsto \bar{l'} \rangle$
	\\
	\hline
	$s_F,s_B,[s_V|{l_1}':n_1,...,{l_m}':n_m],h_B,h_V \models {l_1}' \mapsto {e_1}' * ... * $\\
	$ {l_m}' \mapsto {e_m}'$ and $s_F,s_B,[s_V|{l_1}':n_1,...,{l_m}':n_m],h_B,h_V \models$\\
	$ bk \mapsto ({l_1}',...,{l_m}')$, for some $n_1,...,n_m \in \textrm{Loc}$
	\\
	\hline
	$\textrm{dom}(h_V) = \{n_1,...,n_m\},  h_V(n_1)=\llbracket {e_1}' \rrbracket (s_F,s_B,[s_V|{l_1}':$\\
	$ n_1,...,{l_m}':n_m],h_B) , ... ,h_V(n_m)=\llbracket {e_m}' \rrbracket (s_F,s_B,[s_V|$\\
	${l_1}':n_1,...,{l_m}':n_m],h_B)$, $\textrm{dom}(h_B) = \{\llbracket bk \rrbracket (s_F,s_B,[s_V|$\\ 
	${l_1}':n_1,...,{l_m}':n_m],h_B)\}$, $h_B(\llbracket bk \rrbracket (s_F,s_B,[s_V|{l_1}':n_1,$\\
	$  ...,{l_m}':n_m],h_B))=(n_1,...,n_m)$, for some $n_1,...,n_m \in \textrm{Loc}$
	\\
	\hline
	$\textrm{dom}({h_V}') = \{n_1,...,n_m\} \cup \{n\} $, $h_V(n_1)=\llbracket {e_1}' \rrbracket (s_F,s_B,$\\
	$  [s_V|{l_1}':n_1,...,{l_m}':n_m,l:n],{h_B}') ,...,h_V(n_m)=\llbracket {e_m}' \rrbracket $\\
	$  (s_F,s_B,[s_V|{l_1}':n_1,...,{l_m}':n_m,l:n],{h_B}'), h_V(n_{m+1})=$\\
	$  \llbracket e \rrbracket (s_F,s_B,[s_V|{l_1}':n_1,...,{l_m}':n_m,l:n],{h_B}'), \textrm{dom}({h_B}') =$  \\
	$ \{\llbracket bk \rrbracket (s_F,s_B,[s_V|{l_1}':n_1,...,{l_m}':n_m,l:n],{h_B}')\}  $, and  \\
	${h_B}'(\llbracket bk \rrbracket (s_F,s_B,[s_V|{l_1}':n_1,...,{l_m}':n_m,l:n],{h_B}'))=$\\
	$(n_1,...,n_m,n)$  for some $n_1,...,n_m,n \in \textrm{Loc}$
	\\
	\hline
	$s_F,s_B,[s_V|{l_1}':n_1,...,{l_m}':n_m,l:n],{h_B}',{h_V}' \models \bar{l'} \looparrowright  $  \\
	$\bar{e'} * l \mapsto e$ and $s_F,s_B,[s_V|{l_1}':n_1,...,{l_m}':n_m,l:n],{h_B}',{h_V}'  $ \\
	$\models bk \mapsto \bar{l'} \cdot (l)$ for some $n_1,...,n_m,n \in \textrm{Loc}$
	\\
	\hline
	$s_F,s_B,s_V,{h_B}',{h_V}' \models \exists \bar{l'},l. \langle \bar{l'} \looparrowright \bar{e'} * l \mapsto e, bk \mapsto \bar{l'} \cdot (l) \rangle$
\end{tabular}
\end{scriptsize}
\\[5mm]

\begin{itemize}
	\item The Block Content Lookup form (BCL) 
\end{itemize}

$\llbracket x:=\{bk.e\} \rrbracket \sigma = (s_F,s_B,s_V[h_V(loc_i)/x],h_B,h_V)$

where $ h_B(\llbracket bk \rrbracket \sigma)=(loc_1,...,loc_n) $,$\llbracket e \rrbracket \sigma = i$ and  $1 \leq i \leq n$;

let ${s_V}'' \equiv [s_V|l_1:n_1,...,l_m:n_m]$, ${s_V}' \equiv {s_V}''[h_V(loc_i)/x]$, where $ h_B(\llbracket bk \rrbracket \sigma)=(loc_1,...,loc_n) $,  and  $1 \leq i \leq n$, then:

\begin{scriptsize}
\begin{tabular}{c}
	$\sigma \models \exists \bar{l}. \langle  x=x' \land e=i \land \bar{l} \looparrowright (\bar{e} | i \rightarrowtail x'') , bk \mapsto \bar{l} \rangle$
	\\
	\hline
	$\llbracket x \rrbracket (s_F,s_B,{s_V}'',h_B) = \llbracket x' \rrbracket (s_F,s_B,{s_V}'',h_B)$, $\textrm{dom}(h_V) =$\\
	$ \{n_1,...,n_m\}$, $h_V(n_1) =\llbracket e_1 \rrbracket(s_F,s_B,{s_V}'',h_B),... ,  h_V$\\
	$ (n_{\llbracket e \rrbracket(s_F,s_B,{s_V}'',h_B)}) =\llbracket x'' \rrbracket(s_F,s_B,{s_V}'',h_B), ... , $\\
	$ h_V(n_m) =\llbracket e_m \rrbracket(s_F,s_B,{s_V}'',h_B)$, $\textrm{dom}(h_B)=\{ \llbracket bk \rrbracket$\\
	$ (s_F,s_B,{s_V}'',h_B)\}, ~\text{and}~ h_B(\llbracket bk \rrbracket (s_F,s_B,{s_V}'',h_B))=$\\
	$(n_1,...,n_m)$, for some $n_1,...,n_m \in \textrm{Loc}$  
	\\
	\hline
	$\textrm{dom}(h_V) = \{n_1,...,n_m\}, ~ h_V(n_1) =\llbracket e_1[x'/x] \rrbracket (s_F,s_B,{s_V}', $\\
	$h_B), ... ,h_V(n_{\llbracket e[x'/x] \rrbracket(s_F,s_B,{s_V}',h_B)}) =\llbracket x'' \rrbracket (s_F,s_B,{s_V}', $ \\
	$h_B), ... , h_V(n_m) =\llbracket e_m[x'/x] \rrbracket (s_F,s_B,{s_V}',h_B) $, $\textrm{dom}(h_B)=$\\
	$\{ \llbracket bk[x'/x] \rrbracket (s_F,s_B,{s_V}',h_B)\}, ~\text{and}~ h_B(\llbracket bk[x'/x] \rrbracket (s_F,s_B,$\\
	${s_V}',h_B))= (n_1,...,n_m)$, for some $n_1,...,n_m \in \textrm{Loc}$
	\\
	\hline
	$s_F,s_B,{s_V}',h_B,h_V \models x=x'' \land e[x'/x]=i \land \bar{l} \looparrowright (\bar{e}[x'/x]$\\
	$ | i \rightarrowtail x'')$ and  $s_F,s_B,{s_V}',h_B,h_V \models bk[x'/x] \mapsto \bar{l}$,\\
	 for some $n_1,...,n_m \in \textrm{Loc}$
	\\
	\hline
	$s_F,s_B,{s_V}',h_B,h_V \models \exists \bar{l}. \langle x=x'' \land e[x'/x]=i \land  $\\
	$\bar{l} \looparrowright (\bar{e}[x'/x] | i \rightarrowtail x'') , bk[x'/x] \mapsto \bar{l} \rangle $
\end{tabular}
\end{scriptsize}
\\[5mm]

\begin{itemize}
	\item The Block Address Assignment form (BAA)
\end{itemize}

$\llbracket b:=bk \rrbracket \sigma = (s_F,s_B[\llbracket bk \rrbracket \sigma/b],s_V,h_B,h_V)$

Let ${s_B}' \equiv s_B[\llbracket bk \rrbracket \sigma/b]$, then:

\begin{scriptsize}
\begin{tabular}{c}
	$\sigma \models \langle \alpha , b == b' \land \beta \rangle$
	\\
	\hline
	$\sigma \models \alpha$
	$~\text{and}~\sigma \models b == b' \land \beta$
	\\
	\hline
	$\sigma \models \alpha, ~ \llbracket b \rrbracket \sigma=\llbracket b' \rrbracket \sigma,  $
	$~\text{and}~\sigma \models \beta$
	\\
	\hline
	$s_F,{s_B}',s_V,h_B,h_V \models \alpha[b'/b] \land {s_B}'(b) = \llbracket bk[b'/b] \rrbracket $
	\\
	$  (s_F,{s_B}',s_V,h_B) ~\text{and}~s_F,{s_B}',s_V,h_B,h_V \models \beta[b'/b]$
	\\
	\hline
	$s_F,{s_B}',s_V,h_B,h_V \models \alpha[b'/b] ~\text{and}~ $\\
	$ s_F,{s_B}',s_V,h_B,h_V \models b == bk[b'/b] \land \beta[b'/b]$
	\\
	\hline
	$s_F,{s_B}',s_V,h_B,h_V \models \langle \alpha[b'/b] , b == bk[b'/b] \land \beta[b'/b] \rangle   $
\end{tabular}
\end{scriptsize}
\\[5mm]

\begin{itemize}
	\item The Block Address Assignment form (BAAalt)
\end{itemize}

$\llbracket b:=f.i \rrbracket \sigma = (s_V,s_B[\llbracket f.i \rrbracket \sigma/b],s_F,h_B,h_V)$

Let ${s_B}' \equiv s_B[\llbracket f.i \rrbracket \sigma/b]$, then:

\begin{scriptsize}
\begin{tabular}{c}
	$\sigma \models \langle \#f_2=i-1 \land \alpha,f = f_2 \cdot b' \cdot f_3  \land \beta \rangle$
	\\
	\hline
	$\sigma \models \#f_2=i-1 \land \alpha$
	$~\text{and}~\sigma \models f = f_2 \cdot b' \cdot f_3  \land \beta$
	\\
	\hline
	$\llbracket  \#f_2 \rrbracket \sigma=\llbracket i-1 \rrbracket \sigma, ~\sigma \models \alpha, ~ \llbracket f \rrbracket \sigma=\llbracket f_2 \cdot b' \cdot f_3 \rrbracket \sigma,  $
	$~\text{and}~\sigma \models \beta$
	\\
	\hline
	$\llbracket  b \rrbracket(s_F,{s_B}',s_V,h_B,h_V) = \llbracket  b' \rrbracket(s_F,{s_B}',s_V,h_B,h_V)$
	\\
	$\llbracket  \#f_2 \rrbracket (s_F,{s_B}',s_V,h_B,h_V)=\llbracket i-1 \rrbracket (s_F,{s_B}',s_V,h_B,h_V), ~   $ \\
	$ (s_F,{s_B}',s_V,h_B,h_V) \models \alpha, ~ \llbracket f \rrbracket(s_F,{s_B}',s_V,h_B,h_V)=$\\
	$\llbracket f_2 \cdot b' \cdot f_3 \rrbracket (s_F,{s_B}',s_V,h_B,h_V),  $
	$~\text{and}~(s_F,{s_B}',s_V,h_B,h_V) \models \beta$
	\\
	\hline
	$s_F,{s_B}',s_V,h_B,h_V \models \#f_2=i-1 \land \alpha ~\text{and}~ $\\
	$s_F,{s_B}',s_V,h_B,h_V \models f = f_2 \cdot b' \cdot f_3 \land b==b'  \land \beta$
	\\
	\hline
	$s_F,{s_B}',s_V,h_B,h_V \models \langle \#f_2=i-1 \land \alpha,$\\
	$ f = f_2 \cdot b' \cdot f_3 \land b==b'  \land \beta\rangle   $
\end{tabular}
\end{scriptsize}
\\[5mm]

\begin{itemize}
	\item The Block Address Replacement of a File form (BARF)
\end{itemize}

$\llbracket f.e:=bk\rrbracket \sigma = (s_F[(bloc_1,...,\llbracket bk \rrbracket \sigma,...,bloc_n)/f],s_B,s_V,$\\
$h_B,h_V)$

where  $s_F(f)=(bloc_1,...,bloc_i,...,bloc_n)$,$\llbracket e \rrbracket \sigma = i$ and  $1 \leq i \leq n$

Let ${s_F}' \equiv s_F[(bloc_1,...,\llbracket bk \rrbracket \sigma,...,bloc_n)/f]$, then:

\begin{scriptsize}
\begin{tabular}{c}
	$\sigma \models \langle \#f_2=e-1 \land \alpha ,f = f_2 \cdot bk' \cdot f_3 \land \beta \rangle$
	\\
	\hline
	$\sigma \models \#f_2=e-1 \land \alpha$
	$~\text{and}~\sigma \models f = f_2 \cdot bk' \cdot f_3 \land \beta$
	\\
	\hline
	$\llbracket  \#f_2 \rrbracket \sigma=\llbracket e-1 \rrbracket \sigma, ~ \sigma \models \alpha, ~ \llbracket f \rrbracket \sigma=\llbracket f_2 \cdot bk' \cdot f_3 \rrbracket \sigma,  $
	$~\text{and}~\sigma \models \beta$
	\\
	\hline
	$\llbracket f.e \rrbracket ({s_F}',s_B,s_V,h_B,h_V) = \llbracket  bk \rrbracket({s_F}',s_B,s_V,h_B,h_V)$
	\\
	$\llbracket  \#f_2 \rrbracket ({s_F}',s_B,s_V,h_B,h_V)=\llbracket e-1 \rrbracket ({s_F}',s_B,s_V,h_B,h_V), ~$\\
	$ ({s_F}',s_B,s_V,h_B,h_V) \models \alpha ~\text{and}~ \llbracket f \rrbracket ({s_F}',s_B,s_V,h_B,h_V)=$ \\
	$ \llbracket f_2 \cdot bk \cdot f_3 \rrbracket ({s_F}',s_B,s_V,h_B,h_V)  $
	$~\text{and}~({s_F}',s_B,s_V,h_B,h_V) \models \beta$
	\\
	\hline
	$s_F,{s_B}',s_V,h_B,h_V \models \#f_2=e-1 \land \alpha ~\text{and}~$\\
	$ s_F,{s_B}',s_V,h_B,h_V \models f = f_2 \cdot bk \cdot f_3 \land \beta$
	\\
	\hline
	$s_F,{s_B}',s_V,h_B,h_V \models \langle \#f_2=i-1 \land \alpha,$\\
	$ f = f_2 \cdot b' \cdot f_3 \land b==b'  \land \beta\rangle   $
\end{tabular}
\end{scriptsize}
\\[5mm]

\begin{itemize}
	\item Block Deletion form (BD)
\end{itemize}

$\llbracket \mathbf{delete}~b \rrbracket \sigma = (s_F,s_B,s_V,h_B\rceil(\textrm{dom}(h_B)-\{\llbracket b \rrbracket \sigma\}),h_V)$

let ${h_B}' \equiv h_B\rceil(\textrm{dom}(h_B)-\{\llbracket b \rrbracket \sigma\})$, then:

\begin{scriptsize}
\begin{tabular}{c}
	$\sigma \models \exists \bar{l}. \langle \bar{l} \looparrowright -, b \mapsto \bar{l} \rangle$
	\\
	\hline
	$\textrm{dom}(h_V) = \{n_1,...,n_m\} ~\text{and}~ h_V(n_1) = - , ... ,  h_V(n_m) = - $ and \\ 
	$\textrm{dom}(h_B)=\{ \llbracket b \rrbracket \sigma\} ~\text{and}~ h_B(\llbracket b \rrbracket \sigma)=(n_1,...,n_m)$,\\
	  for some $n_1,...,n_m \in \textrm{Loc}$
	\\
	\hline
	$\textrm{dom}(h_V) = \{n_1,...,n_m\} ~\text{and}~ h_V(n_1) = - , ... ,  h_V(n_m) = - $ and \\
	$\textrm{dom}({h_B}')=\{ \llbracket b \rrbracket \sigma\} - \{\llbracket b \rrbracket \sigma\} =\{\}  $ for some $n_1,...,n_m \in \textrm{Loc}$
	\\
	\hline
	$s_F,s_B,[s_V|l_1:n_1,...,l_m:n_m],{h_B}',h_V \models l_1 \mapsto - * ... * l_m \mapsto -$ 
	\\
	and $s_F,s_B,[s_V|l_1:n_1,...,l_m:n_m],{h_B}',h_V \models \mathbf{emp}_B$,\\
	 for some $n_1,...,n_m \in \textrm{Loc}$
	\\
	\hline
	$s_F,s_B,s_V,{h_B}',h_V \models \exists \bar{l}. \langle \bar{l} \looparrowright -, \mathbf{emp}_B \rangle$
\end{tabular}
\end{scriptsize}


\begin{thebibliography}{1}
	
	\bibitem{H1}
	Apache Software Foundation.: 
	\newblock Rebalance data blocks when new data nodes added or data nodes become full - ASF JIRA, (2019).
	\newblock {https://issues.apache.org/jira/browse/HADOOP-1652}
	
	\bibitem{berdine2005symbolic}
	Josh Berdine, Cristiano Calcagno, and Peter W O'hearn.:
	\newblock Symbolic execution with separation logic.
	\newblock In: {\em Asian Symposium on Programming Languages and Systems}, pages 52--68. Springer, (2005).
	
	\bibitem{birkedal2004local}
	Lars Birkedal, Noah Torp-Smith, and John C Reynolds.:
	\newblock Local reasoning about a copying garbage collector.
	\newblock In: {\em ACM SIGPLAN Notices}, volume 39, pages 220--231. ACM, (2004).
	
	
	\bibitem{brotherston2017biabduction}
	James Brotherston, Nikos Gorogiannis, and Max Kanovich.:
	\newblock Biabduction (and related problems) in array separation logic.
	\newblock In: {\em International Conference on Automated Deduction}, pages 472--490. Springer, (2017).
	
	\bibitem{brotherston2016model}
	James Brotherston, Nikos Gorogiannis, Max Kanovich, and Reuben Rowe.:
	\newblock Model checking for symbolic-heap separation logic with inductive
	predicates.
	\newblock {\em ACM SIGPLAN Notices}, 51(1):84--96, (2016).
	
	\bibitem{dodds2009deny}
	Mike Dodds, Xinyu Feng, Matthew Parkinson, and Viktor Vafeiadis.:
	\newblock Deny-guarantee reasoning.
	\newblock In: {\em European Symposium on Programming}, pages 363--377. Springer, (2009).
	
	\bibitem{dunlop1980comparative}
	Victor R Basili and Douglas D Dunlop.:
	\newblock A comparative analysis of functional correctness.
	\newblock In: {\em ACM Computing Surveys}, volume 14, pages 229--244. ACM, (1980).
	
	
	
	\bibitem{Ghemawat2003gfs}
	S. Ghemawat, H. Gobioff, and S. Leung.:
	\newblock The google file system.
	\newblock In: {\em Proceedings of the 19th ACM Symposium on Operating Systems
		Principles}, 29--43, (2003).
	
	\bibitem{hashem2015rise}
	Ibrahim Abaker Targio Hashem, Ibrar Yaqoob, Nor Badrul Anuar, Salimah Mokhtar,
	Abdullah Gani, and Samee Ullah Khan.:
	\newblock The rise of ``big data'' on cloud computing: Review and open
	research issues.
	\newblock {\em Information Systems}, 47:98--115, (2015).
	
	\bibitem{ishtiaq2001bi}
	Samin S Ishtiaq and Peter W O'hearn.:
	\newblock Bi as an assertion language for mutable data structures.
	\newblock {\em ACM SIGPLAN Notices}, 36(3):14--26, (2001).
	
	\bibitem{james2014program}
	Julian James Stephen, Savvas Savvides, Russell Seidel, and Patrick Eugster.:
	\newblock Program analysis for secure big data processing.
	\newblock In: {\em Proceedings of the 29th ACM/IEEE international conference on
		Automated software engineering}, pages 277--288. ACM, (2014).
	
	\bibitem{Jin2019}
	Zhao Jin, Hanpin Wang, Lei Zhang, Bowen Zhang, Kun Gao, and Yongzhi Cao.:
	\newblock {Reasoning about Block-based Cloud Storage Systems}.
	\newblock {arXiv:1904.04442 [cs.LO]} (2019).
	\newblock {https://arxiv.org/abs/1904.04442}
	
	\bibitem{jing2017modeling}
	Yuxin Jing, Hanpin Wang, Yu~Huang, Lei Zhang, Jiang Xu, and Yongzhi Cao.:
	\newblock A modeling language to describe massive data storage management in
	cyber-physical systems.
	\newblock {\em Journal of Parallel and Distributed Computing}, 103:113--120,
	(2017).
	
	\bibitem{jung2015iris}
	Ralf Jung, David Swasey, Filip Sieczkowski, Kasper Svendsen, Aaron Turon, Lars
	Birkedal, and Derek Dreyer.:
	\newblock Iris: Monoids and invariants as an orthogonal basis for concurrent
	reasoning.
	\newblock In: {\em ACM SIGPLAN Notices}, volume 50, pages 637--650. ACM, (2015).
	
	\bibitem{krogh2017relational}
	Morten Krogh-Jespersen, Kasper Svendsen, and Lars Birkedal.:
	\newblock A relational model of types-and-effects in higher-order concurrent
	separation logic.
	\newblock In: {\em ACM SIGPLAN Notices}, volume 52, pages 218--231. ACM, (2017).
	
	\bibitem{le2017decidable}
	Quang~Loc Le, Makoto Tatsuta, Jun Sun, and Wei-Ngan Chin.:
	\newblock A decidable fragment in separation logic with inductive predicates
	and arithmetic.
	\newblock In: {\em International Conference on Computer Aided Verification},
	pages 495--517. Springer, (2017).
	
	\bibitem{lee2014proof}
	Wonyeol Lee and Sungwoo Park.:
	\newblock A proof system for separation logic with magic wand.
	\newblock In: {\em ACM SIGPLAN Notices}, volume 49, pages 477--490. ACM, (2014).
	
	\bibitem{ntzik2015reasoning}
	Gian Ntzik and Philippa Gardner.:
	\newblock Reasoning about the posix file system: Local update and global
	pathnames.
	\newblock In: {\em ACM SIGPLAN Notices}, volume~50, pages 201--220. ACM, (2015).
	
	\bibitem{o2001local}
	Peter O'Hearn, John Reynolds, and Hongseok Yang.:
	\newblock Local reasoning about programs that alter data structures.
	\newblock In: {\em International Workshop on Computer Science Logic}, pages
	1--19. Springer, (2001).
	
	\bibitem{pereverzeva2013formal}
	Inna Pereverzeva, Linas Laibinis, Elena Troubitsyna, Markus Holmberg, and Mikko
	P{\"o}ri.:
	\newblock Formal modelling of resilient data storage in cloud.
	\newblock In: {\em International Conference on Formal Engineering Methods},
	pages 363--379. Springer, (2013).
	
	\bibitem{pym2018separation}
	David Pym, Jonathan M Spring, and Peter O'Hearn.:
	\newblock Why separation logic works.
	\newblock {\em Philosophy \& Technology}, pages 1--34, (2018).
	
	\bibitem{reynolds2002separation}
	John C Reynolds.:
	\newblock Separation logic: A logic for shared mutable data structures.
	\newblock In: {\em Proceedings 17th Annual IEEE Symposium on Logic in Computer
		Science}, pages 55--74. IEEE, (2002).
	
	\bibitem{shvachko2010hadoop}
	Konstantin Shvachko, Hairong Kuang, Sanjay Radia, and Robert Chansler.:
	\newblock The hadoop distributed file system.
	\newblock In: {\em 2010 IEEE 26th symposium on mass storage systems and
		technologies}, pages 1--10. IEEE, (2010).
	
	\bibitem{svendsen2014impredicative}
	Kasper Svendsen and Lars Birkedal.:
	\newblock Impredicative concurrent abstract predicates.
	\newblock In: {\em European Symposium on Programming Languages and Systems},
	pages 149--168. Springer, (2014).
	
	\bibitem{ta2016automated}
	Quang-Trung Ta, Ton Chanh Le, Siau-Cheng Khoo, and Wei-Ngan Chin.:
	\newblock Automated mutual explicit induction proof in separation logic.
	\newblock In: {\em International Symposium on Formal Methods}, pages 659--676.
	Springer, (2016).
	
	\bibitem{wang2018reasoning}
	Hanpin Wang, Zhao Jin, Lei Zhang, Yuxin Jing, and Yongzhi Cao.:
	\newblock Reasoning about cloud storage systems.
	\newblock In: {\em 2018 IEEE Third International Conference on Data Science in
		Cyberspace}, pages 107--114. IEEE, (2018).
	
	\bibitem{winskel1993formal}
	Glynn Winskel.:
	\newblock {\em The formal semantics of programming languages: an introduction}.
	\newblock MIT press, (1993).
	
	
	\bibitem{yang2001local}
	Yang H.:
	\newblock {\em Local Reasoning for Stateful Programs}.
	\newblock PhD thesis, University of Illinois at Urbana-Champaign, (2001).
	
	
	\bibitem{yang2002semantic}
	Hongseok Yang and Peter O'Hearn.:
	\newblock A semantic basis for local reasoning.
	\newblock In: {\em International Conference on Foundations of Software Science
		and Computation Structures}, pages 402--416. Springer, (2002).
	
	\bibitem{hoare1969axiomatic}
	Hoare, Charles Antony Richard.:
	\newblock An axiomatic basis for computer programming.
	\newblock {\em Communications of the ACM}, pages 576--580, (1969).
	
	\bibitem{dean2008mapreduce}
	Dean, Jeffrey and Ghemawat, Sanjay.:
	\newblock MapReduce: simplified data processing on large clusters.
	\newblock {\em Communications of the ACM}, pages 107--113, (2008).
	
	\bibitem{S3}
	CNN Business.: 
	\newblock Amazon broke the internet with a typo [EB/OL], (2017).
	\newblock {https://money.cnn.com/2017/03/02/technology/amazon-s3-outage-human-error/index.html}
	
	\bibitem{hesselink2012formalizing}
	Hesselink, Wim H and Lali, Muhammad Ikram.:
	\newblock Formalizing a hierarchical file system.
	\newblock {\em Formal Aspects of Computing}, pages 27--44, (2012).
	
\end{thebibliography}
\end{document}